\documentclass[12pt,pra,aps,amssymb,amsfonts,amsmath,tightenlines]{revtex4}
\usepackage{dsfont}
\usepackage{graphicx}
\usepackage{amssymb,amsfonts,amsthm}
\usepackage{color}
\usepackage{verbatim}
\usepackage[normalem]{ulem}
\usepackage{enumerate}
\usepackage{mathrsfs}
\usepackage{bm,float}
\newcommand{\diff}{\mathrm{d}}
\usepackage{textgreek}
\usepackage{lgreek}

\newtheorem{Proposition}{Proposition}

\newtheorem{Lemma}{Lemma}

\usepackage{nicefrac}

\newcommand{\half}{\mbox{$\textstyle \frac{1}{2}$}}

\newcommand{\re}{\mbox{$\rm e$}}

\newcommand{\mcF}{{\mathcal F}}
\newcommand{\mcC}{{\mathcal C}}
\newcommand{\mbP}{{\mathbb P}}
\newcommand{\mbE}{{\mathbb E}}
\newcommand{\mbR}{{\mathbb R}}

\newcommand{\given} {{\,| \,}}

\begin{document}

\title{Information-Based Trading }
\author{George Bouzianis,$^1$ Lane~P.~Hughston,$^1$ and Leandro S\'anchez-Betancourt${}^{2,3}$}

\affiliation{
$^1$Department of Computing, Goldsmiths University of London\\ New Cross, London SE14\,6NW, United Kingdom\\
$^2$Mathematical Institute, University of Oxford\\ Woodstock Road, Oxford OX2\,6GG, United Kingdom\\
$^3$Oxford-Man Institute of Quantitative Finance\\
Walton Well Road, Oxford OX2\,6ED, United Kingdom
}

\begin{abstract}
\noindent 
We consider a pair of traders in a market where the information available to the second trader is a strict subset of the information available to the first trader. The traders make prices based on information concerning a security that pays a random cash flow at a fixed time $T$ in the future. Market information is modelled in line with the scheme of Brody, Hughston \& Macrina (2007, 2008, 2011) and Brody, Davis, Friedman \& Hughston (2009). The risk-neutral distribution of the cash flow is known to the traders, who make prices with a fixed multiplicative bid-offer spread and report their prices to a game master who declares that a trade has been made when the bid price of one of the traders crosses the offer price of the other. We prove that the value of the first trader's position is strictly greater than that of the second. The results are analyzed by use of simulation studies and generalized to situations where (a) there is a hierarchy of traders, (b) there are multiple successive trades, and (c) there is inventory aversion. 
\vspace{-0.2cm}
\\
\begin{center}
{\scriptsize {\bf Key words: Information-based asset pricing, information processes, trading models, informed traders, Brownian bridge,
nonlinear filtering, signal to noise ratio, bid-offer spread, inventory aversion.
} }
\end{center}
\end{abstract}

\maketitle
%
\section{Introduction}
\label{sec: introduction}
\noindent We step outside the standard framework for arbitrage-free pricing and consider a situation that conforms more closely to that of reality, namely that where traders with access to superior information dominate those who are less well informed. The principle that the more well informed should be better placed to trade successfully is intuitively satisfactory and generally accepted, but there remains the problem of embodying this principle in the context of a specific set of models, and within this context giving a precise mathematical characterization of the mechanics of statistical arbitrage when the arbitrage is based on informational advantage. That is the goal of the present investigation.  Our approach to the problem is to adapt the methods of information-based asset pricing \cite{BHM2022} to the dynamical setting of a hierarchy of traders stratified by the level of information that they can access. 

Before embarking upon details we give an overview of the arguments we propose to develop. In Section \ref{sec: the model} we recall the construction of the information-based price in a market with a single source of information concerning a security that delivers a single random cash flow at a fixed time $T$. The price process is presented in Proposition 1 and is worked out by filtering techniques based on the methods of references   \cite{BHM2007, BHM2008}. Here we reformulate that work in a rather general setting in which the cash flow is represented by an integrable random variable. 

The setup is then generalized to the situation where prices are made by a hierarchy of traders, wherein a tier-$n$ trader has at their disposal the filtration generated by a collection of $n$ information processes, the tier-$(n-1)$ trader works with the filtration made by the first $n-1$ of these processes, and so on. In Proposition 2 we work out the price made by a tier-$n$ trader, and in Proposition 3 we show that this price can be expressed in terms of a reduced ``effective" information process with a higher information flow rate given by the square root of the sum of the squares of the flow rates of the various component information sources. 

The arguments of Section \ref{sec: the model} are developed within the framework of a hierarchy of specific models for market information. The restriction to the case of a single random cash flow is merely for simplicity and can be generalized in a straightforward way to the situation where one has multiple cash flows dependent on multiple market factors and multiple associated sources of information  \cite{BHM2007, BHM2008, Macrina2006}.  In Section \ref{sec: deterministic trading time} the discussion is lifted to the general set up of a pair of traders where the filtration of Trader $B$ is a sub-filtration of that of Trader $A$. See \cite{BDFH2009, BHM2011} for similar trading schemes.
In the setting we consider, a {\it game master} monitors the prices made by the two traders and declares the terms under which a trade takes place. 

There is a fixed multiplicative spread $\phi > 1$, and in Scenario 1, the trades take place at a predesignated time and occur whenever the bid price of one of the traders exceeds that of the offer price of the other. There is a pricing kernel that determines a pricing measure, and hence a value can be assigned to Trader $A$'s position as he enters the arena. We assume the bank accounts of both traders are empty at the initiation of trading. Then under a mild non-triviality condition we prove  in Proposition 4 that Trader $A$'s position has a {\it strictly positive value} in this scenario. The method of proof extends that of reference  \cite{BDFH2009}. In Scenario 2 we consider the same trading setup and specialize to the case where Trader $A$ accesses a single information process of the Brownian bridge type and Trader $B$ is an uninformed trader with only the \textit{a priori} distribution of $X$ at his disposal and no updating. In that case we proceed in Proposition 5 to establish an exact lower bound on the value of Trader $A$'s position. 

In Section \ref{sec: trading when the spreads cross} we return to the general setting and consider the situation where a trade takes place on the first occasion over the trading interval $(0, T)$ when the spreads cross, that is, when the bid price of $A$ reaches the offer price of $B$ or the bid price of $B$ reaches the offer price of $A$. The theory of stopping-time $\sigma$-algebras is well suited for this kind of analysis (Lemma 1) and we are able to make novel use of the optional sampling theorem for uniformly integrable martingales for the proof of Proposition 6, which shows that the profitability of the informed trader in such a scenario is strictly positive. Figure 1 shows how the distribution of the trading times is concentrated in the early part of the trading session for smaller spreads and shifts towards the middle part of the trading session for larger spreads. Figure 2 shows the rather surprising result that the informed trader's average profit per trading session is a concave function of the spread factor at a given level of the information flow rate. 

In Section \ref{sec: successive random trades} we allow for the possibility of multiple trades under the assumption that following a trade each trader will take the value at which the trade was carried out to be their new quoted mid-price. Thus the quoted mid-prices are equalized and are made by adjusting the associated information based prices upward or downward by the spread factor. 
The new quoted mid-prices of the two traders then develop onward from there with the information subsequently received, and the associated bid and offer prices are determined by a further unit of the spread factor. Another trade takes place when the spreads cross again. In Scenario 4 we look at the case where the game master allows for trading to continue for up to a maximum of two trades over a given trading session, and then in Proposition 7 we show that the profitability of the informed trader is strictly positive in such a scenario. 

In Section \ref{sec: inventory aversion}, under Scenario 5, we develop a theory of inventory aversion to take into account the fact that in many trading situations traders prefer, everything else being the same, to run fairly flat books, moving in and out of positions in accordance with the arrival of favorable trading opportunities. We allow for this by keeping track of a trader's inventory and adjusting the quoted mid-price up or down by an inventory aversion adjustment factor $\psi$ for each unit of inventory. The two traders can have distinct inventory aversion adjustment factors, reflecting the fact that inventory aversion is an individual characteristic of a trader, unlike the spread factor, which is a market convention determined by the game master. 
Nevertheless there are some global constraints: in particular, in Lemma 2 we show that the inventory aversion risk factors are bounded from above by the spread factor. In the situation where the game master admits a maximum of two trades to take place, our Proposition 8 establishes that even with the presence of inventory aversion (possibly asymmetric) the profitability of the informed trader remains strictly positive. 

In Section \ref{sec: multiple trades} under Scenario 6 we examine the situation when more than two trades are admitted under the condition of inventory aversion. We prove a combinatorial identity in Lemma \ref{Lemma 3}, a result that is of  interest in its own right, illustrated in Figures 3 and 4. This paves the way to Proposition 9, establishing in a rather general setting the positivity of the informed trader's profitability when there are multiple trades. Figure 5 allows one to see in some detail examples of the trading dynamics when there are up to ten trades allowed in a given trading session. We see the bid and offer prices quoted by both traders, and the points at which the spreads cross are flagged with pointers to show whether the trade is buy or a sell. The inventories are tracked in the second panel, and the third panel plots the underlying information-based prices. The fourth panel shows the ratio of the quoted mid prices, and we show the boundaries at which trades are triggered. Then in Figure 6 we plot the average profitability of the informed trader as a function of the spread, which turns out to be concave, illustrating the fact that there is an ideal level for the spread, at a given level of  the information flow, at which the profitability is maximized. Figure 7 then presents this result in the form of a profitability surface parametrized by the spread factor and the information flow rate. Figure 8 shows the effect of inventory aversion and makes it clear that with increased levels of inventory aversion the average value of the magnitude of the inventory decreases. 
Finally, in Section \ref{sec: information dominates strategy} we show that if the trader with inferior information is allowed the strategic flexibility of an adapted spread and an adapted risk aversion, this makes no difference qualitatively to the outcome that the trader with superior information will dominate. We conclude in Section \ref{sec: conclusions}.
\section{The model}
\label{sec: the model}
\noindent We fix a probability space $\left(\Omega,\,\mcF,\,\mbP\right)$ and consider a trading model with stratified information. We look at several different variants of the model with increasing layers of complexity. An interval of time $(0,T)$ is fixed over which the trading takes place. The trading is in a contract that  at time $T$ pays a single non-negative dividend, which we model by an integrable random variable $X$. We assume that interest rates are deterministic and that $\mbP$ is the pricing measure (the risk-neutral measure), so the value of one contract at time $0$ (the present) is given by
\begin{equation}
    S_0 = P_{0T}\,\mbE\left[X\right].
\end{equation}
Here $P_{0t}$ denotes the discount factor out to time $t \in \mathbb R^+$. For simplicity, we assume that traders have the same information and same beliefs at time $0$ and that the position of each trader is initially flat. 

We introduce a hierarchy of traders as follows. By a tier-0 trader, we mean a trader whose only information is that implicit in the set-up described. Thus, he knows the \textit{a priori} measure $\mu_X(\diff x)$ of the random variable $X$ under $\mbP$, but has no other information. As a consequence, the prices made at later times are given by
\begin{equation}
    S_t = \mathds{1}(t<T)\,P_{tT}\,\mbE\left[X\right] ,
\end{equation}
where $P_{tT} = P_{0T}/P_{0t}$. Knowledge of the distribution of $X$ is equivalent to knowledge of the time-$0$ prices of all derivative securities with a payout at time $T$ of the form $\mathds 1\{X\leq x\}$ for some $x \in \mathbb R^+$. Note that $S_t=0$ for $t\geq T$ since the contract pays a single random dividend at $T$ and nothing thereafter. With this convention the price process has the c\`adl\`ag property. 

By a tier-1 trader we mean a trader who has access to a so-called information process $\{\xi_t\}_{t\geq 0}$. The information takes the form
\begin{equation}
    \xi_t = \sigma\,t\,X + \beta_{tT}\,,
\end{equation}
for $0\leq t \leq T$ and $\xi_t = \sigma\,T\,X$ for $t\geq T$. Here $\sigma$ is a non-negative parameter, called the information flow rate, and 
$\{\beta_{tT}\}_{t\geq 0}$ is a Brownian bridge over the interval $[0,T]$.  The Brownian bridge can be modelled by setting 
$   \beta_{tT} = B_t - t \, T^{-1} B_T$
for $0\leq t \leq T$ and $\beta_{tT}=0$ for $t\geq T$, where $\{B_{t}\}_{t\geq 0}$ is a standard Brownian motion.
The information $\xi_t$ available at time $t$ can be thought of as a signal $X$, corresponding to the upcoming cash flow, obscured by market noise. For some applications one can treat $X$ more generally as a market factor with the property that the cash flow is given by a function of $X$. Here for simplicity we stick to the case where $X$ represents the cash flow itself, but the reader will be readily able to see how the more general situation can be accommodated. The theory of information processes along with a variety of applications is set out in references \cite{BHM2007, BHM2008, BDFH2009, BHM2011, BHM2022}; see also \cite{HSB2020, Macrina2006, Rutkowski Yu}. 

Now, let us write $\{\mcF_t\}_{t\geq 0}$ for the smallest right-continuous complete filtration containing the natural filtration of $\{\xi_t\}$. Thus, we let 
${\mathcal N}_{\mathbb P}$ denote the set whose elements are subsets of the null sets of $\mathbb P$, and for each $t\geq0$ we set
\begin{equation}
\mcF_t = \bigcap_{u>t} \sigma [\{\xi_s\}_{s\leq u} ] \, \vee {\mathcal N}_{\mathbb P}\,.
\end{equation}
We refer to $\{\mcF_t\}$ as the filtration generated by $\{\xi_t\}$. Then the price of the asset at time $t$ is
\begin{equation} \label{price general}
    S_t = \mathds{1}(t<T) \,P_{tT}\,\mbE\left[X  \given \mathcal F_t \, \right] ,
\end{equation}
and we have the following: 
\begin{Proposition}
The information-based price of a contract that pays a non-negative integrable random cash flow $X$ at time $T$ takes the form
\begin{equation}
    S_t = \mathds{1}(t<T) \,P_{tT}\,\frac{\int_{x\in \mathbb R^+} x\,\re^{\left(\sigma\,x\,\xi_t - \frac{1}{2}\,\sigma^2\,x^2\,t\right)\,\frac{T}{T-t}}\,\mu_X(\diff x)}{\int_{x\in \mathbb R^+} \re^{\left(\sigma\,x\,\xi_t - \frac{1}{2}\,\sigma^2\,x^2\,t\right)\,\frac{T}{T-t}}\,\mu_X(\diff x)}\,.
\label{price}
\end{equation}
\end{Proposition}
\begin{proof}
One can show that the information process has the Markov property \cite{BHM2007}, from which it follows that 
\begin{equation}
    S_t = \mathds{1}(t<T) \,P_{tT}\,\mbE\left[X \given  \xi_t \, \right] .
\end{equation}
Now write $\{F_{X}(x)\}_{x\in\mathbb R^+}$ for the distribution of $X$, where $F_{X}(x)=\mathbb{P}\left(X\leq x\right)$, and let $Y$ be a continuous random variable with distribution $\{F_{Y}(y)\}_{y\in \mathbb{R}}$ and density $\{f_{Y}(y)\}_{y\in \mathbb{R}}$. Then there is a generalized Bayes theorem (\cite{HSB2020}, Lemma 3) to the effect that for all $y\in\mathbb{R}$ at which $f_{Y}(y)\neq 0$ it holds that
\begin{equation}
F^{(x)}_{X|Y=y}=\frac{\int_{u\in[0,x]}f_{Y|X=u}^{(y)}\,\mu_X(\diff u)}{\int_{u\in[0,\infty)}f_{Y|X=u}^{(y)}\,\mu_X(\diff u)} \, ,
\end{equation}
where $F^{(x)}_{X|Y=y}$ denotes the conditional distribution $\mathbb{P}\left(X\leq x \given Y=y\right)$, and
\begin{equation}
  f_{Y|X=u}^{(y)}=\frac{\diff}{\diff y}\mathbb{P}\left(Y\leq y\given X=u\right)  .
\end{equation}
It follows that for any outcome of chance $\omega \in \Omega$ we have
\begin{equation} \label{eq: generalized Bayes}
\mbE\left[X \mid  Y_{\omega} \, \right]\ = \frac{\int_{x\in \mathbb R^+} x \,f_{Y|X=x}^{(Y_{\omega})}\,\mu_X(\diff x)}{\int_{x\in \mathbb R^+} f_{Y|X=x}^{(Y_{\omega})}\,\mu_X(\diff x)} \, ,
\end{equation}
where $Y_{\omega}$ denotes the value of $Y$ at $\omega \in \Omega$. Then if we let the random variable $\xi_t$ play the role of $Y$ in the setup described above and use the fact that $\xi_t$ is conditionally Gaussian with mean $\sigma t X$ and variance $t(T-t)/T$, it is straightforward to work out the conditional distribution of $X$ given $\xi_t$ by use of \eqref{eq: generalized Bayes}, and \eqref{price} follows. 
\end{proof} 
We refer to $S_t$ as the information-based price made at time $t$ by a tier-1 trader. Note that $S_T=0$, whereas
$ \lim_{t\to T} S_t = X$.
By a higher-tier trader we mean a trader who has one or more additional information processes at his disposal in such a way that the filtration accessible to the tier-$n$ trader is a sub-filtration of that accessible to the tier-$(n+1)$ trader. We adapt the notation above to the case where we have $n$ information processes and write
\begin{equation}
    \xi^i_t = \sigma_i\,\min(t,T)\,X + \beta^i_{tT}
\end{equation}
for $i \in\{1,\,\dots\,, \,n\}$, $t \geq 0$. Here $\sigma_i$ denotes the information flow rate for the information process $\{\xi^i_t\}$, and the $\{\beta^i_{tT}\}_{i=1,\,\dots\,,\,n}$ are independent Brownian bridges. In principle, one could look at the seemingly more general situation where the $n$ Brownian bridges are correlated, but after suitable linear transformations this can be reduced to the case under consideration. We recall that a tier-0 trader has no information available other than that which is available \textit{a priori} to all traders.  A tier-0 trader thus makes the price
\begin{equation}
    S^0_t = \mathds{1}(t<T)\,P_{tT}\,\mbE\left[X\right].
\end{equation}
A tier-1 trader makes the price 
\begin{equation}
    S^1_t = \mathds{1}(t<T)\,P_{tT}\,\mbE\left[X\given \xi^1_t\right],
\end{equation}
a tier-2 trader makes the price 
\begin{equation}
    S^{1:2}_t = \mathds{1}(t<T)\,P_{tT}\,\mbE\left[X\given \xi^1_t,\, \xi^2_t\right],
\end{equation}
and a tier-$n$ trader makes the price 
\begin{equation}
    S^{1:n}_t = \mathds{1}(t<T)\,P_{tT}\,\mbE\left[X\given \xi^1_t,\,\dots,\, \xi^n_t\right].
\end{equation}
Then, it becomes an exercise to prove the following: 
\begin{Proposition}
The tier-$n$ price is given by 
\begin{equation}\label{eq: price 1:n}
    S^{1:n}_t = \mathds{1}(t<T) \,P_{tT}\,\frac{\int_0^\infty x\,\re^{\left(\sigma_1\,x\,\xi^1_t - \frac{1}{2}\,\sigma_1^2\,x^2\,t\right)\,\frac{T}{T-t}}  \,\cdots\, \re^{ \left(\sigma_n\,x\,\xi^n_t - \frac{1}{2}\,\sigma_n^2\,x^2\,t\right)\,\frac{T}{T-t} }\,\mu(\diff x)}{\int_0^\infty \re^{\left(\sigma_1\,x\,\xi^1_t - \frac{1}{2}\,\sigma_1^2\,x^2\,t\right)\,\frac{T}{T-t} } \,\cdots\, \re^{ \left(\sigma_n\,x\,\xi^n_t - \frac{1}{2}\,\sigma_n^2\,x^2\,t\right)\,\frac{T}{T-t} }\,\mu(\diff x)}\,.
\end{equation}
\end{Proposition}
\begin{proof}
It will suffice to look at the case $n=2$ to illustrate the method of proof. We wish to show that 
\begin{equation}\label{eq: price 1:2 to show}
    S^{1:2}_t = \mathds{1}(t<T) \,P_{tT}\,\frac{\int_0^\infty x\,\re^{\left(\sigma_1\,x\,\xi^1_t - \frac{1}{2}\,\sigma_1^2\,x^2\,t\right)\,\frac{T}{T-t} }\,\re^{  \left(\sigma_2\,x\,\xi^2_t - \frac{1}{2}\,\sigma_2^2\,x^2\,t\right)\,\frac{T}{T-t} }\,\mu(\diff x)}{\int_0^\infty \,\re^{\left(\sigma_1\,x\,\xi^1_t - \frac{1}{2}\,\sigma_1^2\,x^2\,t\right)\,\frac{T}{T-t} }\,\re^{  \left(\sigma_2\,x\,\xi^2_t - \frac{1}{2}\,\sigma_2^2\,x^2\,t\right)\,\frac{T}{T-t} }\,\mu(\diff x)}\,.
\end{equation}
We observe that 
\begin{align}
\mbE\left[X \given \xi^1_t,\,\xi^2_t  \right] &= \mbE\left[X \given \sigma_1\,\xi^1_t + \sigma_2\,\xi^2_t ,\,  \sigma_2\,\xi^1_t - \sigma_1\,\xi^2_t  \right] ,
\end{align}
from which one deduces that 
\begin{equation}
 \mbE\left[X \given \xi^1_t,\,\xi^2_t  \right]  = \mbE\left[X \given \xi^{1:2}_t,\,\eta^{1:2}_t  \right] ,
\end{equation}
where 
\begin{equation}
\xi^{1:2}_t = \frac{\sigma_1\,\xi^1_t + \sigma_2\,\xi^2_t }{\sqrt{\sigma^2_1 + \sigma^2_2}}, \quad 
\eta^{1:2}_t =  \frac{\sigma_2\,\xi^1_t - \sigma_1\,\xi^2_t }{\sqrt{\sigma^2_1 + \sigma^2_2}} \,.
\end{equation}
In fact, $\{\xi^{1:2}_t\}_{t\geq 0}$ can be interpreted as the ``effective information" generated jointly by $\{\xi^{1}_t\}_{t\geq 0}$ and $\{\xi^{2}_t\}_{t\geq 0}$. One can check that $\{\xi^{1:2}_t\}$ is an information process and that $\{\eta^{1:2}_t\}$ is an independent Brownian bridge. Indeed, we have 
\begin{equation} \label{eq: information 1:2}
\xi^{1:2}_t = \sigma_{1:2}\,t\,X + \beta^{1:2}_t\,,
\end{equation}
where 
\begin{equation} \label{effective flow rate tier 2}
\sigma_{1:2} = \sqrt{\sigma^2_1 + \sigma^2_2}
\end{equation}
and the process $\{\beta^{1:2}_t\}_{t\geq 0}$ defined by 
\begin{equation}
\beta^{1:2}_t = \frac{\sigma_1\,\beta^1_{tT} + \sigma_2\,\beta^2_{tT}}{\sigma_{1:2}}
\end{equation}
is a standard Brownian bridge over $[0,T]$. The fact that $\{\eta^{1:2}_t\}$ is itself a Brownian bridge and that $\{\xi^{1:2}_t\}$ and $\{\eta^{1:2}_t\}$ are independent can be deduced from covariance relations. Moreover, it should be evident that the $\sigma$-algebras $\sigma\left[\sigma[X],\,\sigma[\{\xi^{1:2}_t\}_{t\geq 0}] \, \right]$ and $\sigma\left[\{\eta^{1:2}_t\}_{t\geq 0}\right]$ are independent. 

Now, we recall the following from  Williams \cite{williams1991}, \S 9.7(k)  in his well-known list of properties of conditional expectation. Suppose that $Z$ is an integrable random variable on a probability space $(\Omega, \mathcal F, \mathbb P)$ and that $\mathcal A$ and $\mathcal B$ are sub-$\sigma$-algebras of $\mathcal F$.  Then if 
$\sigma [\, \sigma[Z], \mathcal A]$ and $\mathcal B$ are independent, we have
\begin{equation}
\mbE\left[Z \given \sigma [ \mathcal A, \mathcal B \,] \, \right] = \mbE\left[Z \given \mathcal A \, \right] .
\end{equation}
It follows in the case under consideration that
\begin{align}
\mbE\left[X \given \xi^1_t,\,\xi^2_t  \right] &= \mbE\left[X \given \xi^{1:2}_t \right] .
\end{align}
That is to say, {\it the tier}-2 {\it price only depends on the effective information associated with the two given information processes}. Thus, we have 
\begin{equation}\label{eq: price 1:2}
    S^{1:2}_t = \mathds{1}(t<T) \,P_{tT}\,\frac{\int_0^\infty x\,\re^{\left(\sigma_{12}\,x\,\xi^{1:2}_t - \frac{1}{2}\,\sigma_{12}^2\,x^2\,t\right)\,\frac{T}{T-t}}\,\mu(\diff x)}{\int_0^\infty \re^{\left(\sigma_{12}\,x\,\xi^{1:2}_t - \frac{1}{2}\,\sigma_{12}^2\,x^2\,t\right)\,\frac{T}{T-t}}\,\mu(\diff x)}\,.
\end{equation}
Then if we substitute \eqref{eq: information 1:2} into \eqref{eq: price 1:2}, a straightforward calculation gives the result \eqref{eq: price 1:2 to show}. The proof for higher values of $n$ is similar.
\end{proof}
It is interesting to note, as we have pointed out above, that the tier-2 price \eqref{eq: price 1:2} can be expressed as a function of a single effective parcel of information, given by \eqref{eq: information 1:2}, with effective flow rate \eqref{effective flow rate tier 2}. This property generalizes to  higher $n$. In particular,  we have:
\begin{Proposition}
The tier-n price of an asset paying $X$ at time $T$ can be put in the form 
\begin{equation}
    S^{1:n}_t = \mathds{1}(t<T) \,P_{tT}\,\frac{\int_0^\infty x\,\re^{\left(\sigma_{1:n}\,x\,\xi^{1:n}_t - \frac{1}{2}\,\sigma_{1:n}^2\,x^2\,t\right)\,\frac{T}{T-t}}\,\mu(\diff x)}{\int_0^\infty \re^{\left(\sigma_{1:n}\,x\,\xi^{1:n}_t - \frac{1}{2}\,\sigma_{1:n}^2\,x^2\,t\right)\,\frac{T}{T-t}}\,\mu(\diff x)}\,,
\end{equation}
where the effective information $(\xi^{1:n}_t)_{t\geq 0}$ is defined by
\begin{equation}
\xi^{1:n}_t = \frac{1}{\sigma_{1:n}}\, \sum_{i=1}^{n} \sigma_i\,\xi^i_t, \quad \sigma_{1:n}^2 = \sum_{i=1}^n \sigma^2_i
\end{equation}
\end{Proposition}
\begin{proof}
We observe that 
\begin{equation}
\mathbb{E}\left[X \given \xi^1_t,\dots ,\xi^n_t \right] = \mathbb{E}\left[X \given  \xi^{1:n}_t,\,\eta^{1:2}_t,\dots , \eta^{1:n}_t \right] ,
\end{equation}
where $\xi^{1:n}_t$ is the effective information at $t$ and the $(\eta_t^{1:i})_{t\geq 0}$, defined for $i\in\{2,\dots, n\}$  by
\begin{equation}
\eta^{1:i}_t = \frac{1}{\sigma_{1:i}}\left( \sigma_i \,\xi^{1:\,i-1}_t - \sigma_{1:\,i-1}\, \xi^i_t\right),
\end{equation}
constitute a family of $n-1$ independent Brownian bridges. For each value of $i\in\{2,\dots, n\}$ the $X$-dependence cancels in the terms on the right and hence we can write
\begin{equation}
\eta^{1:i}_t = \frac{1}{\sigma_{1:i}}\left( \sigma_i \,\beta^{1:\,i-1}_t - \sigma_{1:\,i-1}\, \beta^i_t\right), \quad
\beta^{1:i}_t = \frac{1}{\sigma_{1:i}}\, \sum_{j=1}^{i} \sigma_j\,\beta^j_t.
\end{equation}
The fact that the $(\eta_t^{1:i})$ are indeed independent Brownian bridges can be then checked by the use of covariance relations. The $\sigma$-algebras $\sigma[X,\,\xi^{1:n}_t]$ and $\sigma[\eta^{1:i}_t:i=2,\dots, n ]$ are therefore independent and thus
\begin{equation}
\mathbb{E}\left[X \given \xi^1_t,\,\dots \,,\,\xi^n_t  \right] = \mathbb{E}\left[X \given \xi^{1:n}_t  \right] ,
\end{equation}
as claimed.
\end{proof}

We observe that as a consequence of these relations the effective information flow rate for a tier-$n$ trader takes the form
\begin{equation}
\sigma_{1:n} = \sqrt{\sum_{i=1}^n \sigma^2_i}\, ,
\end{equation}
a relation that might appropriately be referred to as a Pythagorean law of information flow: when there are multiple sources of information, the square of the effective information flow rate is equal to the sum of the squares of the information flow rates of the sources \cite{BDFH2009, BHM2011}. 
\section {Deterministic trading time}
\label{sec: deterministic trading time}
\noindent It will be useful to establish some general principles regarding trading with information. As in the previous section, we consider a market where a single risky security is traded that pays a random dividend $X \geq 0$ at a fixed time $T>0$.  As before, we take interest rates to be deterministic and we let $\mathbb P$ denote the pricing measure. In a more extended treatment we might enter into a discussion of how the pricing measure is determined, but that is not our goal here. 

We look at the situation where there are two traders, a higher-tier Trader $A$ who has the information
$\mathbb F^A = \{\mathcal F^n_t\}$ generated by $\{\xi^{1:n}_t\}_{t\geq 0}$ at his disposal, and a lower-tier Trader $B$ who has the information
$\mathbb F^B = \{\mathcal F^m_t\}$ generated by $\{\xi^{1:m}_t\}_{t\geq 0}$ at his disposal, where $n > m$. Then clearly at each time $t$ we have 
$\mathcal F^B_t \subset \mathcal F^A_t$, reflecting the fact that the higher-tier trader has more information at his disposal than his lower-tier counterparty. 
Equivalently, we say that the filtration $\mathbb F^A$ is finer than the filtration $\mathbb F^B$ and that $\mathbb F^B$ is courser than $\mathbb F^A$.  The key point is that at the outset of the trading the traders have identical information and identical beliefs, both embodied in $\mathbb P$. After time $0$ and up to time $T$ the traders gain information and make prices based on this information in such a way that the information gained by Trader $B$ is always a subset of the information gained by Trader $A$.  

Let us assume that the market conventions are such that traders make prices with a fixed multiplicative spread. We refer to the information-based price $S_t$ computed by a given trader at time $t$ as his mid-price at that time. The spread factor $\phi$ is taken to be a fixed real number that is strictly greater than unity. Then the trader's offer price (the price at which he is willing to sell the asset) will be  $\phi S_t$, and his bid price (the price he is willing to buy the asset at) will be $\phi^{-1} S_t$. 

The so-called bid-offer (or bid-ask) spread is then given by $(\phi - \phi^{-1}) S_t$. The discussion can be easily adjusted to accommodate the case of a fixed additive spread, but we shall leave that to the reader and stick with multiplicative spreads here. The mid-price is then given by the geometric mean of the bid price and the offer price. 

We shall assume in what follows that the market is overseen by a {\it game master} (e.g., an exchange or central planner). The traders report their mid-prices to the game master (but not to one another) on a continuous basis, and the game master determines, in accordance with certain conventions, when a trade is deemed to have taken place. Neither trader is aware of informational status of the other trader. Trader $A$ receives the information that Trader $B$ receives, but he is not aware of the fact that this is the only information being received by Trader $B$. There are various situations that can be envisaged, and we refer to these as {\it trading scenarios}, several examples of which we proceed to examine.

\vspace{.3cm}
\noindent {\bf Scenario 1}. The game master selects a fixed time $t \in (0, T)$ at which a trade can take place. The traders each initially have flat positions. We shall write
\begin{equation}\label{eq: S^A}
S^A_t =  \mathds{1}(t<T)\,P_{tT}\, \mathbb{E}\left[X \given \mathcal F^A_t \right]
\end{equation}
for the mid-price made by Trader $A$ at time $t$, and let $S^{A-}_t = \phi^{-1}S^A_t$ and $S^{A+}_t = \phi S^A_t$ denote his bid price and offer price, respectively. For the mid-price of Trader $B$ at $t$, we write
\begin{equation}\label{eq: S^B}
S^B_t =  \mathds{1}(t<T)\,P_{tT}\, \mathbb{E}\left[X \given \mathcal F^B_t \right] ,
\end{equation}
with bid price $S^{B-}_t = \phi^{-1}S^B_t$ and offer price $S^{B+}_t = \phi S^B_t$. The game master declares that Trader $A$ sells one unit of the asset to Trader $B$ at time $t$ if $S^{A+}_t \leq  S^{B-}_t$;  the game master declares that Trader $A$ buys one unit of the asset from Trader $B$ at time $t$ if 
$S^{A-}_t \geq  S^{B+}_t$; if neither of these conditions hold then no trade takes place. By convention, the game master assigns the mean price 
 $(S^{A+}_t S^{B-}_t)^{1/2}$ to the trade if $A$ sells to $B$; the game master assigns the mean price  $(S^{A-}_t S^{B+}_t)^{1/2}$ to the trade if $A$ buys from $B$. In fact, these prices are the same, both being equal to the geometric mean of the mid-prices made by $A$ and $B$. The overall profit (or loss) $H^A_T$ accruing to Trader $A$ at time $T$ is thus given by 
\begin{align}
 H^A_T = \left [(P_{tT})^{-1}(S^{A}_t S^{B}_t)^{1/2} - X \right] \mathds 1(S^{A+}_t \leq S^{B-}_t) + 
  \left [ X - (P_{tT})^{-1}(S^{A}_t S^{B}_t)^{1/2} \right] \mathds 1(S^{A-}_t \geq  S^{B+}_t) ,
\end{align}
where payments made or received at $t$ are future-valued to $T$. The profit (or loss) of Trader $B$ is clearly $ H^B_T =  -H^A_T$. Now, since $\mathbb P$ is the risk neutral measure, it is meaningful to assign an overall value to Trader $A$'s strategy, given by  
\begin{equation} \label{value at time 0}
H^A_0 = P_{0T} \, \mathbb{E}\left[H^A_T \right] .
\end{equation}
Here we have taken the $\mathbb P$-expectation of Trader $A$'s profit when that profit is expressed in units of the money-market account at time $T$, which in a deterministic interest-rate system is given by $(P_{0T})^{-1}$. Armed with  \eqref{value at time 0}, we are now in a position to say more specifically in what sense the additional information possessed by Trader $A$ gives him a financial advantage.  To this end, let us say that a trading model is \textit{non-trivial} if the probability of a trade occurring over a given trading session is non-vanishing. Note that we do not require that a trade should definitely occur, merely that it might occur. Since the pricing measure and the physical measure share the same null sets, it follows that a trading model is non-trivial if and only if the probability of a trade occurring in any given trading session is non-zero under the measure $\mathbb P$. Then we have: 
\begin{Proposition}
The value of Trader A's position under Scenario 1 is strictly positive in any non-trivial trading model. 
\end{Proposition}
\begin{proof} If $A$ sells then  $S^{A+}_t \leq S^{B-}_t$ and hence $(S^{A}_t S^{B}_t)^{1/2} = (S^{A+}_t S^{B-}_t)^{1/2} \geq S^{A+}_t = \phi\, S^{A}_t $ ; whereas if $A$ buys then $S^{A-}_t \geq  S^{B+}_t$ and hence $(S^{A}_t S^{B}_t)^{1/2} = (S^{A-}_t S^{B+}_t)^{1/2} \leq S^{A-}_t = \phi^{-1} S^{A}_t  $. It follows that
\begin{equation} \label{inequality on payoff}
 H^A_T \geq \left [(P_{tT})^{-1}\phi \,S^{A}_t  - X \right] \mathds {1}(S^{A+}_t \leq  S^{B-}_t) + 
  \left [ X - (P_{tT})^{-1}\phi^{-1}S^{A}_t  \right] \mathds {1}(S^{A-}_t \geq S^{B+}_t) \, .
\end{equation}
Then if we take into account the fact that $\mathcal F^B_t \subset \mathcal F^A_t$ and hence that $S^{B}_t$ is 
$\mathcal F^A_t$-measurable, as are the two indicator functions in the expression above, it follows by use of \eqref{eq: S^A}, \eqref{eq: S^B}, \eqref{value at time 0}, \eqref{inequality on payoff} and the tower property that
\begin{equation}\label{eq: strictly positive bound}
H^A_0 \geq P_{0t} \,  \mathbb{E}\left[ (\phi - 1) S^{A}_t \mathds {1}(S^{A+}_t \leq S^{B-}_t) + (1 - \phi ^{- 1})S^{A}_t \mathds {1}(S^{A-}_t \geq S^{B+}_t) \right] ,
\end{equation}
which gives us a strictly positive lower bound on the value of Trader $A$'s position, providing that at least one of the two indicator functions takes the value unity on a set that is not of measure zero, which is the condition that the model is non-trivial.
\end{proof} 
The idea of the proof of Proposition 4 is deceptively simple. We only require that $S^{B}_t$ should be $\mathcal F^A_t$-measurable. If the information available to $A$ is marginally better than that available to $B$, then their mid-prices will not tend to diverge much, and the probability of a trade occurring will be relatively low; whereas if the information available to $A$ is of significantly better quality, then the prices will tend more to diverge, and hence the probability of a trade occurring will be higher, and the price at which the trade occurs will tend to be more advantageous to $A$ than $B$. It should be emphasized that in setting the trading time to be a fixed time $t$ we are not suggesting that such a scenario is representative of the way in which real trades are carried out. The point is rather that under that simplifying assumption we are able to identify how Trader $A$ can take advantage of his informational superiority. By stripping away some of the clutter associated with more realistic trading scenarios, we can focus on the mathematical principle that allows the informationally superior trader to succeed. Of course, we must check that the advantage is not spoiled in the setting of more elaborate trading scenarios, and this is what we do in subsequent examples. We observe that the value of Trader $A$'s position at the beginning of the trading session is strictly positive. This means that he can in principle {\it monetize} his advantage by selling off some or all of his stake in the eventual outcome of the trading session to a third party at the market price given by \eqref{value at time 0}. The fact that such a monetization is feasible is what we mean by statistical arbitrage. There is no guarantee that in any particular trading session Trader $A$ will prevail. 
\vspace{.10cm}

\noindent {\bf Scenario 2}.
Continuing with a deterministic trading time, we remark that in the situation where $A$ is a tier-1 trader and $B$ is a tier-0 trader (only having knowledge of the \emph {a priori} distribution of $X$) one can work out the lower bound on Trader $A$'s position exactly if we use the model of Section 2 when the payout is that of a defaultable discount bond, thus allowing for an explicit realization of the effects of the spread and the flow rate on the profitability of the informed trader's market activity. 
In the case of a defaultable discount bond with no recovery we assume that $X$ takes the values $1$ (no default) and $0$ (default) with probabilities $p$ and $1-p$ respectively. Then the mid-prices made by the two traders are given by
\begin{equation}\label{eq: midprices}
    S^A_t = \mathds{1}(t<T) \,P_{tT}\,\frac{p\,\re^{\left(\sigma\,\xi_t - \frac{1}{2}\,\sigma^2\,t\right)\,\frac{T}{T-t}}}
    {p\,\re^{\left(\sigma\,\xi_t - \frac{1}{2}\,\sigma^2\,t\right)\,\frac{T}{T-t}} + (1-p)}\,, \quad S^B_t = \mathds{1}(t<T) \,P_{tT}\,p \, .
\end{equation}
A straightforward calculation shows that
\begin{equation} \label{eq: first indicator}
\mathds {1}(S^{A+}_t \leq  S^{B-}_t) = \mathds {1}\left( \xi_t \leq \frac{\,T-t}{\sigma T} \log \frac {1-p} {\phi^2 - p} + \half \sigma t \right) 
\end{equation}
and
\begin{equation}\label{eq: second indicator}
\mathds {1}(S^{A-}_t \geq  S^{B+}_t) = \mathds {1}\left( \xi_t \geq \frac{\,T-t}{\sigma T} \log \frac {\, \phi^2(1-p)} {1 - \phi^2 p\,} + \half \sigma t \right) .
\end{equation}
The indicator function  in \eqref{eq: second indicator} is able to take the value unity only if $ \phi^2 p < 1$. This is because under Scenario 2 a trade can take place where Trader $A$ is a buyer only if this condition is satisfied. Otherwise Trader $A$ can only be a seller. In the following we let $N: \mathbb R \to [0,1]$ denote the standard normal distribution function. 

%
\begin{Proposition}
The lower bound on the value of Trader $A$'s position under Scenario 2 is 
\begin{align} \label{eq: exact bound} 
H^A_0  \geq   
\,p \, P_{0T}\, &(\phi - 1) \,N \left[  \frac{1}{\sigma} \sqrt{ \frac{T-t}{t\,T} } \log \left(\frac{1-p}{\phi^2 - p}\right) - \half \sigma \sqrt{ \frac{T\,t}{T-t} } \,\right]
\nonumber\\
& + p \, P_{0T}\, (1 - \phi^{-1}) \,N \left[  \frac{1}{\sigma} \sqrt{ \frac{T-t}{t\,T} } \log \left(\frac{1-\phi^2p}{\phi^2( 1- p)}\right) + \half \sigma \sqrt{ \frac{T\,t}{T-t} } \, \right] .
\end{align}
\end{Proposition}
%
\begin{proof}
We need to work out the right hand side of \eqref{eq: strictly positive bound} when we substitute \eqref {eq: midprices},  \eqref{eq: first indicator} and  \eqref{eq: second indicator} into it. This gives
\begin{align}\label{eq: H_0 prop 5}
& H^A_0   \geq  
\,p \, P_{0T}\, (\phi - 1) \, \mathbb{E}\left[ \, \frac{\,\re^{\left(\sigma\,\xi_t - \frac{1}{2}\,\sigma^2\,t\right)\,\frac{T}{T-t}}}
    {p\,\re^{\left(\sigma\,\xi_t - \frac{1}{2}\,\sigma^2\,t\right)\,\frac{T}{T-t}} + (1-p)}\, \mathds {1}\left( \xi_t \leq \frac{\,T-t}{\sigma T} \log \frac {1-p} {\phi^2 - p} + \half \sigma t \right) \right]
\nonumber\\
& \, \, \,+ p \, P_{0T}\, (1 - \phi^{-1}) \, \mathbb{E}\left[ \,\frac{\,\re^{\left(\sigma\,\xi_t - \frac{1}{2}\,\sigma^2\,t\right)\,\frac{T}{T-t}}}
    {p\,\re^{\left(\sigma\,\xi_t - \frac{1}{2}\,\sigma^2\,t\right)\,\frac{T}{T-t}} + (1-p)}\, \mathds {1}\left( \xi_t \geq \frac{\,T-t}{\sigma T} \log \frac {\, \phi^2(1-p)} {1 - \phi^2 p\,} + \half \sigma t \right) \, \right]  .
\end{align}
Here we have used the fact that $P_{0t}\, P_{tT}=P_{0T}$ in a deterministic interest system. It may not be obvious how to take the expectations in \eqref{eq: H_0 prop 5}, but these can be worked out in closed form. The trick is to observe that the process $\{\Phi_t\}_{0 \leq t < T}$ defined by
\begin{equation}
\Phi_t = \frac{1} {p\,\exp{\left(\sigma\,\xi_t - \frac{1}{2}\,\sigma^2\,t\right) \frac{T}{T-t}} + (1-p)}
\end{equation}
is a unit-initialized martingale under $\mathbb P$ with respect to the filtration generated by $\{\xi_t\}_{0 \leq t < T}$, and hence can be used to change measure to a new measure  
$\mathbb P^*$ under which $\{\xi_t\}_{0 \leq t < T}$ is a 
Brownian bridge. In fact, $\mathbb P^*$ is an example of the so-called bridge measure introduced in reference \cite{BHM2007}. Thus for any event $A_t \in \mathcal F_t$ we define
\begin{equation}
\mathbb P^*(A_t ) = \mathbb E [\, \Phi_t \,\mathds 1 (A_t) \,]\, ,
\end{equation}
and we obtain
\begin{align}
& H^A_0   \geq  
\,p \, P_{0T}\, (\phi - 1) \, \mathbb{E}^*\left[ \, \re^{\left(\sigma\,\xi_t - \frac{1}{2}\,\sigma^2\,t\right)\,\frac{T}{T-t}}
    \, \mathds {1}\left( \xi_t \leq \frac{\,T-t}{\sigma T} \log \frac {1-p} {\phi^2 - p} + \half \sigma t \right) \,\right]
\nonumber\\
& \, \, \,+ p \, P_{0T}\, (1 - \phi^{-1}) \, \mathbb{E}^*\left[ \, \re^{\left(\sigma\,\xi_t - \frac{1}{2}\,\sigma^2\,t\right)\,\frac{T}{T-t}}
    \, \mathds {1}\left( \xi_t \geq \frac{\,T-t}{\sigma T} \log \frac {\,\phi^2(1-p)} {1 - \phi^2 p\,} + \half \sigma t \right) \, \right]  ,
\end{align}
where the expectation is now taken under $\mathbb P^*$. Finally, we use that fact that $\xi_t$ is a $\mathbb P^*$ normally distributed random variable with mean zero and variance $t(T-t)/T$, and the problem reduces to evaluating a Gaussian integral of a standard type, leading to \eqref{eq: exact bound}.
\end{proof} 

\section{Trading when the spreads cross} 
\label{sec: trading when the spreads cross}
\noindent {\bf Scenario 3}. We consider a setup like that of Scenario 1, but where the game master declares that a trade occurs  the first time the spreads cross. To make the ideas precise we need a few results from the theory of stopping times (see, e.g., \cite{Dothan1990, Protter2004, Cohen2015}). We recall that by a stopping time on a filtered probability space $(\Omega, \mcF, \{\mcF_t\}_{t\geq 0},\mbP)$ we mean a random variable $\tau$ taking values in $[0,\infty]$ with the property that $\{\tau\leq t\}\in\mcF_t$ for all $t\geq 0$. It follows that if $\tau_1$ and $\tau_2$ are $\{\mcF_t\}$ stopping times, then $\min(\tau_1,\tau_2)$ is also an $\{\mcF_t\}$ stopping time. This is because $\{\min(\tau_1,\tau_2) \leq t\} = \{\tau_1\leq t\}\cup \{\tau_2 \leq t\}$. 
We recall that in our model the filtration $\{\mathcal F^A_t\}$ accessed by Trader $A$ is strictly finer than the filtration $\{\mathcal F^B_t\}$ accessed by Trader $B$. Thus the mid-prices made both by Trader $A$ and Trader $B$ are adapted to  $\{\mathcal F^A_t\}$. Now, for any càdlàg process $\{Z_t\}_{t\geq 0}$ adapted to $\{\mcF_t\}_{t\geq 0}$, and for any Borel set $D\in\mathcal{B}[0,\infty]$ the random variable $\tau=\inf_{t\in[0,\infty]}(Z_t\in D)$ is an $\{\mcF_t\}$ stopping time. 
We can therefore introduce the following $\{\mathcal F^A_t\}$ stopping times: 
\begin{equation}
\tau^+ = \inf_{t \in [0, T]} (\phi^{-1}S^{A}_t \geq  \phi \,S^{B}_t), \quad \tau^- = \inf_{t \in [0, T]} (\phi\,S^{A}_t \leq  \phi^{-1}S^{B}_t)\,,
\end{equation}
and we set $\tau = \tau^+ \land  \tau^-$, the minimum of $\tau^+$ and  $\tau^-$. Here for any index set $\Lambda \subset \mbR$ and any collection 
$\{A_\lambda\}_{\lambda \in \Lambda}$ of elements of $\mathcal F$ we write $\inf_{\lambda \in \Lambda} \{A_\lambda\}$ for the random variable
$\inf \{ \lambda \in \Lambda : \mathds 1 (\omega \in A_\lambda) = 1\}$. We observe that $ \tau^+$ is the first time the bid price of $A$ hits the offer price of $B$, that $\tau^-$ is the first time the offer price of $A$ hits the bid price of $B$, and that $\tau$ is the first time the spreads cross. We recall that the prices made by both traders drop to zero at time $T$. Therefore if $\tau$ takes the value $T$ for some outcome of chance, then no trade has taken place in that outcome of chance. 

We assume that once a trade has been made then the position will be held until maturity and no further trades will take place. With this setup in mind, we see that the overall profit $H^A_T$ accruing to Trader $A$ at time $T$ in Scenario 3 is 
\begin{align}\label{eq: H^A_T scn3}
 H^A_T =  \left [ X - P_{\tau^+T}^{-1}\,\phi^{-1} S^{A}_{\tau^+} \right] \mathds 1(\tau^+ < \tau^-) + 
  \left [P_{\tau^-T}^{-1}\, \phi\,S^{A}_{\tau^-} - X \right]\mathds 1(\tau^- < \tau^+) .
\end{align}
As in Scenario 1, we can work out the value of Trader $A$'s position by use of the risk-neutral pricing relation \eqref{value at time 0}. For this purpose, we need a few additional mathematical tools. For any stopping time $\tau$ on a filtered probability space $(\Omega, \mcF, \{\mcF_t\}_{t\geq 0},\mbP)$, the so-called stopping-time $\sigma$-algebra $\mcF_{\tau}$ is defined by 
\begin{equation}\label{stopping time sigma algebra}
\mcF_{\tau} = \left\{ A\in \mcF \,:\, A\cap \{\tau \leq t\}\in\mcF_t \,,\, \forall t\geq 0 \right\}\,.
\end{equation}
Given a pair of stopping times $\tau_1$ and $\tau_2$, one can show that  if $\tau_1<\tau_2$ then $\mcF_{\tau_1}\subset \mcF_{\tau_2}$. The following is fundamental: 
%
%
\begin{Lemma} \label{Lemma 1}
If $\alpha$ and $\beta$ are stopping times based on the filtration $\{\mcF_t \}_{t \geq 0}$, then each of the events $\alpha < \beta$, $\alpha \leq \beta$, $\alpha = \beta$, $\alpha \geq \beta$, $\alpha >\beta$ belongs to $\mcF_{\alpha}$ and to $\mcF_{\beta}$\,. 
\end{Lemma} 
\begin{proof} Let us demonstrate that $\{\alpha < \beta \} \in \mcF_{\beta}$. The other cases then follow analogously. Thus by \eqref{stopping time sigma algebra} we need to show that
$\{\alpha < \beta \}\cap \{\beta < t\}\in\mcF_t \, , \, \forall t\geq 0 $\,. Now, 
\begin{equation}
\{\alpha < \beta \} = \bigcup_{q \in \mathbb Q} \{\alpha < q\} \cap \{ q < \beta \} \, ,
\end{equation}
and hence for any $t \geq 0$ it holds that
\begin{equation}
\{\alpha < \beta \}\cap \{\beta \leq t\} = \bigcup_{q \in \mathbb Q} \{\alpha < q \leq t\} \cap \{ q < \beta  \leq t\} \, ,
\end{equation}
for $q$ rational. But each of the events $\{\alpha < q \leq t\}$ and $\{ q < \beta  \leq t\}$ is in $\mcF_t$\,. Thus we have expressed $\{\alpha < \beta \}\cap \{\beta \leq t\} $ as a countable union of sets in $\mcF_t$.
\end{proof}
Continuing with our discussion of stopping times, we recall that a collection $\mcC$ of random variables on a probability space $(\Omega, \mcF, \mbP)$  is said to be uniformly integrable (UI) if given $\epsilon >0$ $\exists \,\delta\in[0,\infty)$ such that $\mbE\left[|Z|\,\mathds 1 (|Z|>\delta)\right]<\epsilon$, $\forall Z\in\mcC$. A martingale $\{M_t\}_{t\geq 0}$ on a filtered probability space $(\Omega, \mcF, \{\mcF_t\}_{t\geq 0},\mbP)$ is said to be closed by a random variable $Y$ if $\mbE[|Y|]<\infty$ and $M_t = \mbE[Y\,|\, \mcF_t]$ for $0\leq t <\infty$. If a right-continuous martingale $\{M_t\}$ is  UI, then $M_\infty=\lim_{t\to \infty} M_t$ exists almost surely, $\mbE[|M_\infty|]<\infty$, and $M_\infty$ closes $\{M_t\}$.
In what follows, we need the optional sampling theorem, which states that if a right-continuous martingale $(M_t)$ is closed by a random variable $M_\infty$, and if $\tau_1$ and $\tau_2$ are stopping times such that $\tau_1\leq \tau_2$ almost surely, then $M_{\tau_1}$ and $M_{\tau_2}$ are integrable and 
\begin{equation}
\mbE\left[M_{\tau_2}\mid \mcF_{\tau_1}\right]= M_{\tau_1}\,.
\end{equation}
This relation holds in particular in the case of a UI martingale. These results can be applied as follows: 
%
\begin{Proposition}
The value of Trader A's position under Scenario 3 is strictly positive in any non-trivial trading model. 
\end{Proposition} 
\begin{proof}
\noindent We observe that if $(S_t)_{t\geq 0}$ is the price process of an asset that pays a single dividend $X$ at time $T$, then the so-called deflated gain process $(G_t)_{t\geq 0}$ defined by 
\begin{equation} \label{deflated gains}
G_t = P_{0t}\,S_t + P_{0T}\, \mathds 1 (t\geq T) \,X
\end{equation}
is a $\mbP$-martingale. 
The deflated gain process is obtained by taking the current value of the asset, expressed in units of the money market account, and adding to it the cumulative dividend process, where each dividend is expressed in units of the money-market account at the time the dividend is paid. In the case of a single dividend payment and a deterministic interest rate system, the result is given by \eqref{deflated gains}.  Thus, by  \eqref{eq: S^A} it holds, in fact, that
\begin{equation}
G_t = P_{0T} \, \mbE\left[X \given \mcF^A_t\right] ,
\end{equation}
and we see that the deflated gain process is a UI martingale under $\mathbb P$, closed by $P_{0T}\,X$. 
It follows by the optional sampling theorem that 
\begin{equation}
 {G}_{\tau^{\pm}} = \mbE\left[ G_T\given \mcF_{\tau^{\pm}} \right] 
\end{equation}
and hence
\begin{equation} \label{optional sampling}
(P_{\tau^{\pm}T})^{-1}\, {S}_{\tau^{\pm}} + X\,  \mathds 1 (\tau^\pm = T) = \mbE\left[ X \given \mcF_{\tau^{\pm}} \right] .
\end{equation}
Then, since $\mathds 1(\tau^+<\tau^-)$ is $\mcF_{\tau^+}$ measurable and $\mathds 1(\tau^-<\tau^+)$ is $\mcF_{\tau^-}$ measurable we can use  the tower property  alongside \eqref{eq: H^A_T scn3} and \eqref{optional sampling} to deduce that 
\begin{align} \label{value of scenario 2 at time 0}
 H^A_0 =&\, P_{0T} (1 - \phi^{-1}) \mathbb{E} \left [P_{\tau^+T}^{-1} \,S^{A}_{\tau^+}  \mathds 1(\tau^+ < \tau^-)\right] \nonumber\\
 &\quad\quad\quad + 
  P_{0T}  (\phi - 1) \mathbb{E} \left [P_{\tau^-T}^{-1}\,S^{A}_{\tau^-} \mathds 1(\tau^- < \tau^+)\right].
\end{align}
It should be evident that both of the terms on the right side of \eqref{value of scenario 2 at time 0} are non-negative. Then since by assumption
it holds that $\mathbb P  (\tau^+ < \tau^-) > 0$ or $ \mathbb P  (\tau^- < \tau^+) > 0$,  or equivalently $\mathbb P  (\tau < T) > 0$, we deduce that $H^A_0> 0$. 
\end{proof}

Now, the information-based model described in Section 2 based on Brownian bridge information is evidently non-trivial, so Proposition 6 is applicable, and we should be able to work out the profitability of Trader $A$ by use of simulation studies. In Figure \ref{fig: trading times}, we show the distribution of trading times over the interval $[0,T]$ under Scenario 3. The information flow rates for both information processes are set at unity. Charts are shown for four different values of the spread factor. 

One observes that for relatively low spreads, e.g.~$\phi = 1.02$, the bulk of the trades occur relatively early on in the trading session, whereas as the spread is increased to higher levels such as $\phi = 1.10$ the trades tend to take place later in the session. This is because it takes more time on average in the case of large spreads for the prices to diverge sufficiently for the spreads to cross. 
%
%
\begin{figure}[H]
    \centering
    \includegraphics[height=3.5in]{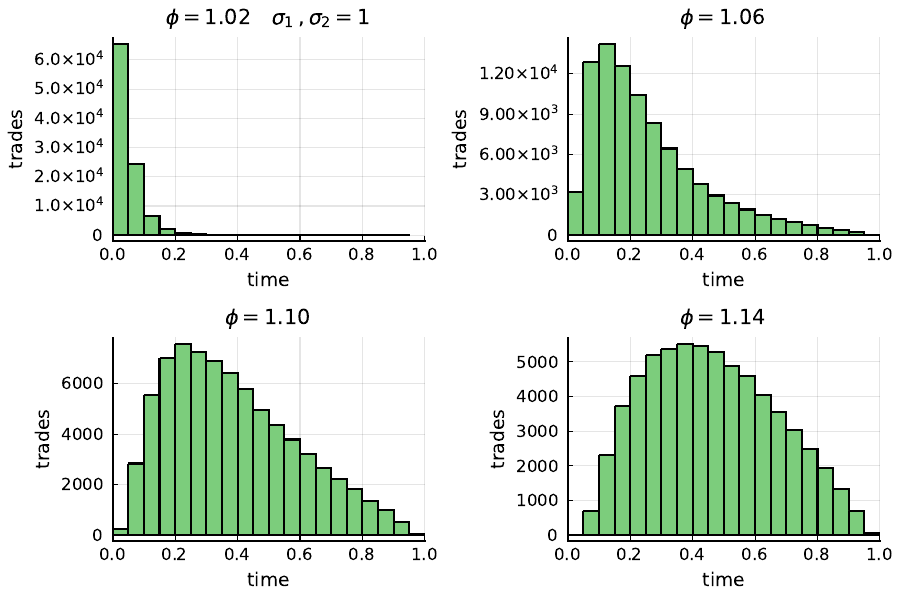}
        \caption{Distribution of trading times under Scenario 3 based on 100,000 simulations, including both buys and sells.
      In this case $r=0$ and $p=0.8$.}
    \label{fig: trading times}
\end{figure} 
In Figure 2 below we plot a heat map showing the average profits taken by Trader $A$ under Scenario 3 as a function of the information rate and the spread factor. We look both at average per-trade profits and average per-session profits. In the latter case, we allow for the fact that a trade may or may not actually occur in a given trading session. We consider 100,000 trading sessions in each case, using the same outcomes of chance. The average profit is plotted as a function of the spread ($x$-axis) and the information rate ($y$-axis). 

We note that the per-trade profits is an increasing function of  the spread factor. On the other hand, the per-session profit is for each value of the information flow rate a concave function of the spread factor, over the range of parameters considered. This allows us to conclude that from the perspective of the better informed trader there is an optimal market spread at which to be trading for any given level of the information flow rate, if the objective is to optimize the per-session profitability. 

These conclusions are of course based on a relatively simple trading model, but nevertheless give a useful qualitative picture of the interplay of information and market convention in the determination of trading profits. Our philosophy is not to present the most elaborate trading models possible, with bells and whistles, but rather the simplest versions of the models that illustrate the underlying mathematical principles that guarantee the success of the more well-informed trader. 
\vspace{.5cm}
%
%
\begin{figure}[H]
    \centering
    \includegraphics[height=2in]{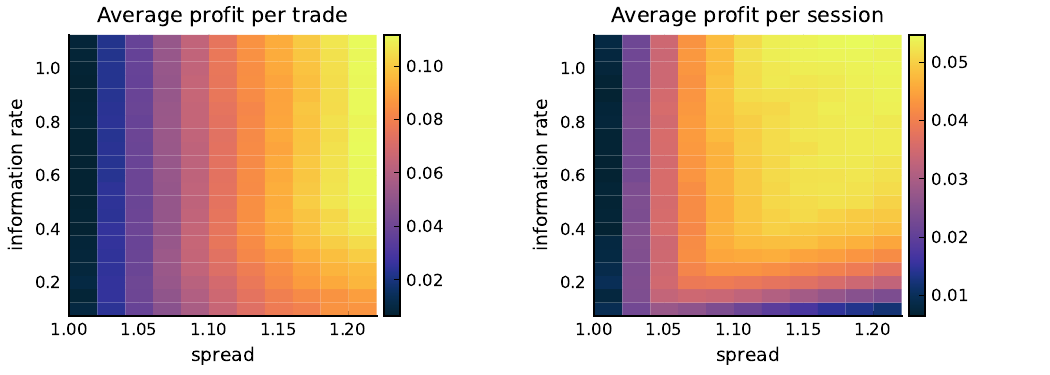}
        \caption{Heat chart of average profit as a function of the spread and the information flow rate. We look at per-trade profits on the left and per-session profits on the right. Here $r=0$ and $p=0.8$. }
    \label{fig: average profit}
\end{figure} 
%
%
\section{successive random trades} 
\label{sec: successive random trades}
\noindent {\bf Scenario 4}. In this scenario we consider the situation where trading occurs when spreads cross, as in Scenario 3, but where additionally prices are adjusted following each trade in such a way that the mid prices are equalized and trades occur when the spreads cross again. The setup is as follows. As before, we consider a contract that pays a  random dividend $X\geq 0$ at time $T$. Let 
$(S^A_t)$ and $(S^B_t)$ denote the mid-prices computed by Trader $A$ and Trader $B$ on the basis of the information they have gained, respectively,  from the information processes at their disposal, and let $\phi>1$ be the spread factor. The filtration available to Trader $B$ is assumed to be a strict sub-filtration of that available to Trader $A$.  Then we define a collection of $\{\mathcal F^A_t\}$ stopping times as follows. First we write
\begin{equation}
\tau_1^+ = \inf_{t \in [0, T]} (\phi^{-1}S^{A}_t \geq  \phi S^{B}_t)\,, \quad \tau_1^- = \inf_{t \in [0, T]} (\phi S^{A}_t \leq  \phi^{-1}S^{B}_t)\,,
\end{equation}
and set $\tau_1 = \tau_1^+ \land \tau_1^-$. Thus $\tau_1$ is the time at which the first trade occurs, with the understanding that $\tau_1=T$ corresponds to the situation where no trade takes place in the given trading session. Trader $A$ buys if $\tau_1^+<\tau_1^-$ and Trader $A$ sells if $\tau^-< \tau_1^+$. Next we introduce an indicator function $Q_1$ for Trader $A$ taking the values plus one, minus one, or zero, to tell us whether the trade was a buy or sell or if there was no trade. We set
\begin{equation}
Q_1 = \mathds 1(\tau_1^+ < \tau_1^-) - \mathds 1(\tau_1^- < \tau_1^+)\, .
\end{equation}
After a trade has taken place, the seller adjusts his mid-price up by one unit of the spread factor, whereas the buyer adjusts his price down by one unit of the spread factor. In this way, the two new mid-prices are equalized. One might envisage that the traders would adjust their prices by taking into account the additional information gained by knowledge of the fact that another trader must have made a price in such a way that the spreads would cross. But since neither trader has knowledge of the other's informational status, it is not so obvious how this could be achieved. We shall assume that it is simply a market convention that the prices are adjusted in line with the procedure stated. Thus each trader adjusts his mid-price to the price at which the trade just took place. This is not unreasonable. Once the prices have been equalized, the traders create bid and offer prices by multiplying or dividing the new mid-prices by a further unit of the spread factor. The game master then declares the time $\tau_2 = \tau_2^+ \land \tau_2^-$ at which the second trade occurs, where
\begin{equation}
\tau_2^+ = \inf_{t \in [\tau_1, T]} (\phi^{-Q_1-1}S^{A}_t \geq  \phi ^{Q_1+1}S^{B}_t)\,, 
\quad \tau_2^- = \inf_{t \in [\tau_1, T]} (\phi^{-Q_1 + 1} S^{A}_t \leq  \phi^{Q_1-1}S^{B}_t)\,,
\end{equation}
with the understanding that $\tau_2=T$ corresponds to the situation where there is no second trade. The profits made by Trader $A$ over the course of the first two trades are given in total by the following expression: 
\begin{align}
 &H^A_T = \left [ X - P_{\tau_1T}^{-1} \,\phi^{-1}S^{A}_{\tau_1}\, \right] \mathds 1(\tau_1^+ < \tau_1^-)  + 
  \left [P_{\tau_1T}^{-1} \,\phi \,S^{A}_{\tau_1} - X \right] \mathds 1(\tau_1^- < \tau_1^+)  \nonumber \\
  &\, \, \, +  \left [ X - P_{\tau_2T}^{-1}\, \phi^{-Q_1 -1}S^{A}_{\tau_2} \, \right] \mathds 1(\tau_2^+ < \tau_2^-) + 
 \left [P_{\tau_2T}^{-1} \,\phi ^{-Q_1 + 1}\,S^{A}_{\tau_2} - X \right] \mathds 1(\tau_2^- < \tau_2^+) \, .
\end{align}
By use of the tower property, the optional sampling theorem, and measurability properties of the stopping times involved, we deduce that the value of Trader $A$'s position is 
\begin{align}  
 H^A_0 = &P_{0T} (1 - \phi^{-1}) \mathbb{E} \left [P_{\tau_1T}^{-1} \,S^{A}_{\tau_1} \, \mathds 1(\tau_1^+ < \tau_1^-)\right]+ 
  P_{0T}  (\phi - 1) \mathbb{E} \left [P_{\tau_1T}^{-1}\,S^{A}_{\tau_1}\, \mathds 1(\tau_1^- < \tau_1^+)\right] \nonumber \\
  & + P_{0T} (1 - \phi^{-2} ) \, \mathbb{E} \left [P_{\tau_2T}^{-1} \,S^{A}_{\tau_2} \, \mathds 1(\tau_1^+ < \tau_1^-) \, \mathds 1 (\tau_2^+ < \tau_2^-)\right] \nonumber \\ 
  & \quad + P_{0T}  (\phi^{2} - 1) \,  \mathbb{E} \left [P_{\tau_2T}^{-1}\,S^{A}_{\tau_2}\,  \mathds 1(\tau_1^- < \tau_1^+) \, \mathds 1(\tau_2^- < \tau_2^+)\right] .
\end{align}
One notes that the terms involving a pair of opposite trades, that is to say, a buy followed by a sell, or a sell followed by a buy, generate no value for Trader $A$, and hence do not appear in the formula above. But the single trade positions, and the positions involving two buys or two sells,  have strictly positive value, assuming that the events in question take place with nonzero probability. Thus we arrive at: 
\vspace{.25cm} 
\begin{Proposition}
The value of Trader A's position under Scenario 4 is strictly positive in any non-trivial trading model. 
\end{Proposition} 
%
\section{Inventory Aversion} 
\label{sec: inventory aversion}
\noindent {\bf Scenario 5}. We consider the situation where traders prefer to keep a  low profile in the market and thus wish to avoid taking excessively long or short positions. This is not unusual: it is often the case that traders are market-neutral and tend to keep their position flat on average. This situation can be modelled as follows. Suppose in the context of the general setup of the last two sections that a trade has just been undertaken, and the new mid-price now shared by both traders is $\bar S_{\tau_1}$ where $\tau_1$ is the time of the trade. Thus $\bar S_{\tau_1} = \phi^{-1} S^A_{\tau_1}$  if the initial trade was a purchase by Trader $A$, and $\bar S_{\tau_1} = \phi S^A_{\tau_1}$ if the initial trade was a sale by Trader $A$. Now let $Q^A_{\tau_1}$ be the inventory of Trader $A$ at time $\tau_1$. By inventory, we mean the number of contracts held by $A$. Then  the inventory of Trader $B$ is given by  $Q^B_{\tau_1} = - Q^A_{\tau_1}$ since the initial positions of each trader are assumed to have been flat. We shall assume now  that at the time the new mid-price is determined there is a further small adjustment of the mid-price set in such a way as to discourage the development of excessively large long or short positions -- that is to say, to keep the absolute inventory relatively small. The required adjustment, which holds with effect from time $\tau$, takes the form 
\begin{equation}
\bar S_{\tau_1} = \phi ^{-Q^A_{\tau_1}} S^A_{\tau_1} \to \phi ^{-Q^A_{\tau_1}} \psi_A^{-Q^A_{\tau_1}} S^A_{\tau_1} \, ,
\end{equation}
for Trader $A$, where the adjustment factor $\psi_A$ is taken to be strictly greater than unity. Thus, if the inventory is positive (long position), the mid-price is knocked down by a factor $\psi_A^{-1}$, hence lowering both the bid price and the offer price a bit. This will tend to suppress further purchases by the trader and will encourage sales. But if the inventory is negative (short position), the mid-price is bumped up by a factor of $\psi_A$, and this will raise the bid price and the offer price, encouraging purchases by the trader and discouraging sales.  

We have assumed so far that traders have the same initial beliefs and initial knowledge, and hence the same initial mid-price; at this stage we can allow for the possibility that they have different levels of inventory aversion. Then Trader $B$ adjusts his mid-price after the first trade by the prescription
\begin{equation}
\bar S_{\tau_1} = \phi ^{Q^B_{\tau_1}} S^B_{\tau_1} \to \phi ^{Q^B_{\tau_1}} \psi_B^{Q^B_{\tau_1}} S^B_{\tau_1} \, ,
\end{equation}
where $\psi_B > 1$ is the inventory aversion adjustment factor for Trader $B$. 

In fact, one can derive an upper bound for the inventory aversion factors. The argument is as follows. Let us assume that at some given time in the trading session Trader $A$ buys ({\it  resp.}, sells) the contract paying ({\it  resp.}, receiving) the amount $\phi^{-1}S^A_t$ ({\it  resp.}, $\phi S^A_t$). Then, the trader's new mid price becomes $\phi^{-1}\psi_A^{-1} S^A_t $ ({\it  resp.}, $\phi\psi_A S^A_t $). Now, this new mid-price has to satisfy the condition that the associated new offer price $\psi_A^{-1} S^A_t $ ({\it  resp.}, bid price $\psi_A S^A_t $) based on it, should be strictly greater than ({\it  resp.}, strictly less than) the price just paid ({\it resp.}, received). In other words, $\psi_A$ has to satisfy an inequality of the form $\psi_A^{-1} S^A_t >\phi^{-1} S^A_t $ ({\it  resp.}, $\psi_A S^A_t <\phi S^A_t $). For it would be irrational, or at least inefficient, in the absence of further information, for Trader $A$ to be willing to immediately sell the contract for less than or the same as what he has just bought it ({\it  resp.}, buy the contract for more than or the same as what he has just sold it). For such actions would at the very least lead to some loss of shoe leather and in general would lead to opportunists taking advantage of him by arbitrage. Hence we obtain: 
%
%
\begin{Lemma}
 The inventory aversion adjustment factors $\psi_A$ and  $\psi_B$ are strictly bounded from above by the spread factor $\phi$. 
 \end{Lemma}

With this scheme in mind we return to the setup of Scenario 4 but we allow now for inventory aversion in our Scenario 5. The profit made by Trader $A$ on the first trade remains the same as in Scenario 4, whereas the inventory aversion adjustments begin to take effect in relation to the prices made in anticipation of a second trade. In particular we now set
\begin{equation}
\tau_2^+ = \inf_{t \in [\tau_1, T]} (\phi^{-Q^A_1-1}\psi_A^{-Q^A_1}S^{A}_t \geq  \phi ^{-Q^B_1+1}\psi_B^{-Q^B_1}S^{B}_t)\,,
\end{equation}
and
\begin{equation}
\tau_2^- = \inf_{t \in [\tau_1, T]} (\phi^{-Q^A_1 + 1} \psi_A^{-Q^A_1}S^{A}_t \leq  \phi^{-Q^B_1-1}\psi_B^{-Q^B_1}S^{B}_t), 
\end{equation}
where for convenience we write $Q^A_1$ for $Q^A_{\tau_1}$ and $Q^A_1$ for $Q^B_{\tau_1}$.
Assuming that trading stops after two trades, we see that the profit made by Trader $A$ when both traders have inventory aversion takes the form
\begin{align}
 H^A_T =&   \left [ X - P_{\tau_1T}^{-1} \,\phi^{-1}S^{A}_{\tau_1}\, \right]\mathds 1(\tau_1^+ < \tau_1^-)  + 
 \left [P_{\tau_1T}^{-1}\, \phi \,S^{A}_{\tau_1} - X \right] \mathds 1(\tau_1^- < \tau_1^+)  \nonumber \\
  &\, \, \,+\left [ X - P_{\tau_2T}^{-1}\, \phi^{-Q^A_1 -1}\, \psi_A^{-Q^A_1}\, S^{A}_{\tau_2} \, \right]  \mathds 1(\tau_2^+ < \tau_2^-) \nonumber \\ 
  & \quad \, \, \, + \left [P_{\tau_2T}^{-1}\, \phi ^{-Q^A_1 + 1}\, \psi_A^{-Q^A_1}\,S^{A}_{\tau_2} - X \right] \mathds 1(\tau_2^- < \tau_2^+) \, .
\end{align}
The value of Trader $A$'s position can be worked out by the methods already discussed, and the result is as follows: 
\begin{align} \label{eq: H^A_0 inventory 2 trades - 1st}
 H^A_0 = \, &P_{0T} (1 - \phi^{-1}) \mathbb{E} \left [P_{\tau_1T}^{-1} \,S^{A}_{\tau_1} \, \mathds 1(\tau_1^+ < \tau_1^-)\right]+ 
  P_{0T}  (\phi-1) \mathbb{E} \left [P_{\tau_1T}^{-1}\,S^{A}_{\tau_1}\, \mathds 1(\tau_1^- < \tau_1^+)\right] \nonumber \\
  & + P_{0T} (\phi^2 \psi_A - 1) \, \mathbb{E} \left [P_{\tau_2T}^{-1} \,S^{A}_{\tau_2} \, \mathds 1(\tau_1^- < \tau_1^+) \, \mathds 1 (\tau_2^- < \tau_2^+)\right] \nonumber \\ 
  & \quad + P_{0T} (\psi^{-1}_A - 1) \, \mathbb{E} \left [P_{\tau_2T}^{-1} \,S^{A}_{\tau_2} \, \mathds 1(\tau_1^+ < \tau_1^-) \, \mathds 1 (\tau_2^- < \tau_2^+)\right] \nonumber \\ 
   & \quad \quad + P_{0T} (1 - \psi_A) \, \mathbb{E} \left [P_{\tau_2T}^{-1} \,S^{A}_{\tau_2} \, \mathds 1(\tau_1^- < \tau_1^+) \, \mathds 1 (\tau_2^+ < \tau_2^-)\right] \nonumber \\ 
  &\quad \quad  \quad + P_{0T}  (1 - \phi^{-2} \psi^{-1}_A) \,  \mathbb{E} \left [P_{\tau_2T}^{-1}\,S^{A}_{\tau_2}\,  \mathds 1(\tau_1^+ < \tau_1^-) \, \mathds 1(\tau_2^+ < \tau_2^-)\right] .
\end{align}
Here the first two terms represent the profits from the first trade, whereas the remaining terms represent the profits from the second trade, allowing for the different possible ways in which the history of the second trade might evolve. 
\begin{Proposition}
The value of Trader A's position under Scenario 5 is strictly positive in any non-trivial trading model with inventory risk. 
\end{Proposition} 
%
\begin{proof}
First we use the optional sampling theorem together with Lemma \ref{Lemma 1} to show that
\begin{align}\label{eq: tower prop 8}
\mathbb{E} \left [P_{\tau_1T}^{-1} \,S^{A}_{\tau_1}  \, \mathds 1(\tau_1^{\pm} < \tau_1^{\mp})\right] &= 
\mathbb{E} \left [\mathbb{E} \left [  X \given \mcF_{\tau_1} \right]  \mathds 1(\tau_1^{\pm} < \tau_1^{\mp})\right] \nonumber \\ &=
\mathbb{E} \left [ \mathbb{E} \left [  \mathbb{E} \left [  X \given \mcF_{\tau_1} \right]  \mathds 1(\tau_1^{\pm} < \tau_1^{\mp}) \given \mcF_{\tau_2}\right] \right] \nonumber \\ &=
\mathbb{E} \left [ \mathbb{E} \left [  \mathbb{E} \left [  X \given \mcF_{\tau_1} \right]  \given \mcF_{\tau_2} \right]  \mathds 1(\tau_1^{\pm} < \tau_1^{\mp}) \right] \nonumber \\ &=
\mathbb{E} \left [\mathbb{E} \left [  X \given \mcF_{\tau_2} \right]  \mathds 1(\tau_1^{\pm} < \tau_1^{\mp})\right] \nonumber \\ &=
\mathbb{E} \left [P_{\tau_2T}^{-1} \, S^{A}_{\tau_2} \, \mathds 1(\tau_1^{\pm} < \tau_1^{\mp})\right] .
\end{align}
Then if we substitute the relation
\begin{align}
\mathds 1(\tau_1^{\pm} < \tau_1^{\mp}) =&\, \mathds 1(\tau_1^{\pm} < \tau_1^{\mp}) \mathds 1(\tau_2=T)\nonumber\\
& \quad + \mathds 1(\tau_1^{\pm} < \tau_1^{\mp}) \mathds 1(\tau_2^+ < \tau_2^-) +
\mathds 1(\tau_1^{\pm} < \tau_1^{\mp}) \mathds 1(\tau_2^- < \tau_2^+)
\end{align}
into \eqref{eq: H^A_0 inventory 2 trades - 1st} and make use of \eqref{eq: tower prop 8}, we obtain
\begin{align} \label{eq: H^A_0 inventory 2 trades}
 H^A_0 = \,
  & P_{0T} (\phi^2 \psi_A + \phi - 2) \, \mathbb{E} \left [P_{\tau_2T}^{-1} \, S^{A}_{\tau_2} \, \mathds 1(\tau_1^- < \tau_1^+) \, \mathds 1 (\tau_2^- < \tau_2^+)\right] \nonumber \\ 
  & \quad + P_{0T} (\psi^{-1}_A - \phi^{-1}) \, \mathbb{E} \left [P_{\tau_2T}^{-1} \, S^{A}_{\tau_2} \, \mathds 1(\tau_1^+ < \tau_1^-) \, \mathds 1 (\tau_2^- < \tau_2^+)\right] \nonumber \\ 
   & \quad \quad + P_{0T} (\phi - \psi_A) \, \mathbb{E} \left [P_{\tau_2T}^{-1} \, S^{A}_{\tau_2} \, \mathds 1(\tau_1^- < \tau_1^+) \, \mathds 1 (\tau_2^+ < \tau_2^-)\right] \nonumber \\ 
  &\quad \quad  \quad + P_{0T}  (2 - \phi^{-1} -  \phi^{-2} \psi^{-1}_A) \,  \mathbb{E} \left [P_{\tau_2T}^{-1}\, S^{A}_{\tau_2}\,  \mathds 1(\tau_1^+ < \tau_1^-) \, \mathds 1(\tau_2^+ < \tau_2^-)\right] \nonumber \\
  &\quad\quad\quad\quad+ P_{0T} (1 - \phi^{-1}) \mathbb{E} \left [P_{\tau_1T}^{-1} \, S^{A}_{\tau_1} \, \mathds 1(\tau_1^+ < \tau_1^-)\, \mathds 1(\tau_2 = T)\right] \nonumber\\
  &\quad\quad\quad\quad\quad + 
  P_{0T}  (\phi-1) \mathbb{E} \left [P_{\tau_1T}^{-1}\, S^{A}_{\tau_1}\, \mathds 1(\tau_1^- < \tau_1^+)\, \mathds 1(\tau_2 = T)\right] ,
\end{align}
and by Lemma 2 it follows that the coefficients of all six terms are strictly positive. 
\end{proof}

\section{Multiple trades}
\label{sec: multiple trades}
\noindent {\bf Scenario 6}. When multiple successive trades take place in a given trading session one can adapt the notation of the previous sections by recursively defining the stopping times
\begin{equation}
\tau_k^+ = \inf_{t \in [\tau_{k-1}, T]} (\phi^{-Q^A_{k-1}-1}\psi_A^{-Q^A_{k-1}}S^{A}_t \geq  \phi ^{-Q^B_{k-1}+1}\psi_B^{-Q^B_{k-1}}S^{B}_t)\,,
\end{equation}
and
\begin{equation}
\tau_{k}^- = \inf_{t \in [\tau_{k-1}, T]} (\phi^{-Q^A_{k-1} + 1} \psi_A^{-Q^A_{k-1}}S^{A}_t \leq  \phi^{-Q^B_{k-1}-1}\psi_B^{-Q^B_{k-1}}S^{B}_t)\, , 
\end{equation}
where $k=1,\dots,n$, with the convention that $\tau_0=0$. Here, the inventory $Q^A_k$ of Trader $A$ after the first $k$ trades is given by 
\begin{equation}
    Q^A_k = \sum_{r=1}^k  \mathds 1(\tau_r^+ < \tau_r^-) - \mathds 1(\tau_r^- < \tau_r^+)\,,
\end{equation}
and it should be evident that $Q^B_k = - Q^A_k$. 
With this notation in mind, let us consider the situation where the game master permits up to $n$ trades in the trading session. We introduce a collection of indices $\epsilon_k$ for $k=1,\dots,n$ that take the values $\pm 1$. Thus, for  $k\in\{1,2,\dots,n\}$ the index $\epsilon_k$ ranges over the set 
$\{+1,-1\}$. 
To analyze the profitability of Trader $A$ over the given trading session it will be useful to have a compact expression for the trading profits. We note that in the case of a single trade, the profit can be written 
\begin{align} \label{single trade}
 H^A_T (1)=   \sum_{\epsilon_1} \epsilon_1 \left(X - \phi^{- \epsilon_1} S_{\tau_1}P^{-1}_{\tau_1T} \right)\, 
    \big[\, \half(1+\epsilon_1) \mathds 1(\tau_1^+ < \tau_1^-) + \half(1-\epsilon_1) \mathds 1(\tau_1^- < \tau_1^+) \, \big] ,
\end{align}
and it can be verified that this reduces to formula \eqref{eq: H^A_T scn3}, with which we are already familiar. A little less obviously, one can check that for 
$n = 2$ we have
\begin{align} \label{double trade}
& H^A_T (2)=   \sum_{\epsilon_1} \epsilon_1 \left(X - \phi^{- \epsilon_1} S_{\tau_1} P_{\tau_1T}^{-1}\right)\, 
    \big[\, \half(1+\epsilon_1) \mathds 1(\tau_1^+ < \tau_1^-) + \half(1-\epsilon_1) \mathds 1(\tau_1^- < \tau_1^+) \, \big] \nonumber \\
  &+    \sum_{\epsilon_1, \epsilon_2} \epsilon_2 \left(X - \phi^{- (\epsilon_1 + \epsilon_2)}  \psi^{-\epsilon_1} S_{\tau_2} P_{\tau_2T}^{-1} \right) 
 \prod_{k=1,2} \big[\, \half(1+\epsilon_k) \mathds 1(\tau_k^+ < \tau_k^-) + \half(1-\epsilon_k) \mathds 1(\tau_k^- < \tau_k^+) \, \big] .
\end{align}
Here we have separated the profits resulting from the first trade from the profits resulting from the second trade. Alternatively, we an write
\begin{align} \label{separated, compact}
& H^A_T (2)= 
H^A_T (1) \, \mathds 1(\tau_2 = T)    
   +  \sum_{\epsilon_1, \epsilon_2} \bigg[ \bigg( \epsilon_1 \left(X - \phi^{- \epsilon_1} S_{\tau_1} P_{\tau_1T}^{-1}\right)
+    \epsilon_2 \left(X - \phi^{- (\epsilon_1 + \epsilon_2)}  \psi^{-\epsilon_1} S_{\tau_2} P_{\tau_2T}^{-1} \right) \bigg) \nonumber \\
& \quad \quad \quad  \quad \times \prod_{k=1,2} \big[\, \half(1+\epsilon_k) \mathds 1(\tau_k^+ < \tau_k^-) + \half(1-\epsilon_k) \mathds 1(\tau_k^- < \tau_k^+) \, \big]
\bigg] ,
\end{align}
which splits the profits into those deriving from the situation where there is a single trade and those deriving from situations where there are two trades. 
The advantage of \eqref{separated, compact} is that this expression, when taken with \eqref{single trade},  readily generalizes to the $n$-trade case. To this end, it will be useful to define the random variables
\begin{align}
\mathscr{I}_k (\epsilon_k) =   \half(1+\epsilon_k) \mathds 1(\tau_k^+ < \tau_k^-) + \half(1-\epsilon_k) \mathds 1(\tau_k^- < \tau_k^+) \, ,
\end{align}
for $k=1,\dots,n$. Thus, $\mathscr{I}_k (1) = \mathds 1(\tau_k^+ < \tau_k^-)$ and $\mathscr{I}_k (-1) = \mathds 1(\tau_k^- < \tau_k^+)$. Let us write
\begin{equation}
    \sum_{\epsilon_1,\dots,\epsilon_n} = \sum_{\epsilon_1}\,\sum_{\epsilon_2}\,\dots \,\sum_{\epsilon_n} .
\end{equation}
One can then check that the total profitability of Trader $A$ over the given time frame in a multiple trade situation is given by
\begin{align}\label{eq: n trades H^A_T}
H^A_T(n) & = H^A_T(n-1)  \, \mathds 1(\tau_n = T) \nonumber \\ & + \sum_{\epsilon_1,\dots,\epsilon_n} \left[ \sum_{k=1}^n \epsilon_k\left(X - \phi^{-\sum_{r=1}^k \epsilon_r}\psi^{-\sum_{r=1}^{k-1} 
\epsilon_r} S_{\tau_k} P_{\tau_k T}^{-1} \right)\,\prod_{k=1}^{n} \mathscr{I}_k (\epsilon_k) \right].
\end{align}
\noindent The profits registered in \eqref{eq: n trades H^A_T} are those for the first $n$ trades, allowing for the possibility that there may be fewer than $n$ trades. The first term on the right takes into account the situation where there are $n-1$ (or fewer) trades, and the second term on the right gives the profitability when there are exactly $n$ trades. 
Here, to ease the notation we have written $\psi$ in place of $\psi_A$, since $\psi_B$ enters only indirectly, via the stopping times. As a step towards establishing Proposition \ref{prop: multiple trades} below, we introduce the following:
%
\begin{Lemma}\label{Lemma 3}
Let $\phi > \psi \geq 1$, let $n\in \mathbb{N}$, and let the series $(\epsilon_k)_{k = 1,2,\dots,n}$ be chosen such that $\epsilon_k\in\{1,-1\}$ for each $k\in\{1,\dots, n\}$. Then 
\begin{equation}\label{eq: lemma ineq}
    \sum_{k=1}^{n} \epsilon_k\left(1 - \phi^{-\sum_{r=1}^k \epsilon_r}\psi^{-\sum_{r=1}^{k-1} \epsilon_r} \right) > 0\,.
\end{equation}
\end{Lemma}
\begin{proof}
We proceed by induction. It is straightforward to check that \eqref{eq: lemma ineq} holds for $n=1$ and for $n=2$. Our goal is to show that if \eqref{eq: lemma ineq} holds for $n-1$ then it holds for $n$. Let the index series $(\epsilon_k)_{k\in\{1,\dots,n\}}$ be given. Define the series $(Q_k)_{k\in\{0,1,\dots,n\}}$ by setting $Q_0=0$ and $Q_k=Q_{k-1}+\epsilon_k$ for $k = 1,\dots,n $. We observe that \eqref{eq: lemma ineq} can be written
\begin{equation}\label{eq: lemma ineq 2}
    \sum_{k=1}^{n} \epsilon_k\left(1 - \phi^{-Q_k}\psi^{-Q_{k-1}} \right)>0\,.
\end{equation}
We consider the case  $\epsilon_1 = 1$, the case $\epsilon_1=-1$ being  analogous. Within the chosen case, we observe that  if a negative value is  taken by an element of the series $(Q_k)_{k\in\{1,\dots,n\}}$, then $\exists\, m \in\{2,\dots,n-1\}$ such that $Q_m=0$ and $Q_{m+1}=-1$. 
Define the series $(R_k)_{k\in\{0,\dots,n-m\}}$ by $R_0 = 0$ and $R_k = R_{k-1}+\epsilon_{k+m}$ for $k\in\{1,\dots,n-m\}$. Then we have
\begin{align}
    &\sum_{k=1}^{n} \epsilon_k\left(1 - \phi^{-Q_k}\psi^{-Q_{k-1}} \right) \nonumber\\
    &\quad = \sum_{k=1}^{m} \epsilon_k\left(1 - \phi^{-Q_k}\psi^{-Q_{k-1}} \right) + \sum_{k=1}^{n-m} \epsilon_{m+k}\left(1 - \phi^{-R_k}\psi^{-R_{k-1}} \right)\,.\label{eq: ineq induction}
\end{align}
Since $m<n$ and $n-m<n$, it holds by the inductive hypothesis that both terms on the right hand side of  \eqref{eq: ineq induction} are positive. Thus, it suffices to restrict our attention to the case where all elements of the series $(Q_k)_{k\in\{1,\dots,n\}}$ are non-negative. Now let $N<n$ be the number of times $\epsilon_k$ is negative for $k = 1, 2, \dots , n$. That is to say,
\begin{equation}
    N = |\{k\in\{1,\dots,n\} : \epsilon_k = -1\}|\,,
\end{equation}
where $|\{\cdot\}|$ denotes the cardinality of the set $\{\cdot\}$. If $N=0$, then $\epsilon_k = 1$ for all $k\in\{1,\dots,n\}$ and \eqref{eq: lemma ineq} holds. Thus, we turn to the case  $N\geq 1$. Let $\alpha_1$ denote the first $k$ for which $\epsilon_k$ is negative, let $\alpha_2$ denote the second $k$ for which $\epsilon_k$ is negative, and so on up to $\alpha_N$, which denotes the $N$-th $k$ for which $\epsilon_k$ is negative. Let $\beta_1$  be $\alpha_1 -1$, and for $k\in\{2,\dots,N\}$ set
\begin{equation}\label{eq: zeta_k}
    \beta_k = \sup \big\{ j \in \{1,\dots,\alpha_k\} \setminus \{\beta_1,\dots,\beta_{k-1} \}  : Q_j>Q_{\alpha_k} \big\} .
\end{equation}
It should be evident that for each $\alpha_k$ such that $k\in\{1,\dots,N\}$ there exists a number $\beta_k$ satisfying \eqref{eq: zeta_k}. To fix ideas, in Figure \ref{fig: lemma 3} we present an example of the trajectory of Trader $A$'s inventory  $(Q_k)_{k\in\{0,\dots,6\}}$ over six trades, where we work out $\alpha_k$ and $\beta_k$ for $k\in\{1,2\}$.
\vspace{.25cm} 
\begin{figure}[H]
    \centering
    \includegraphics[height=2.2in]{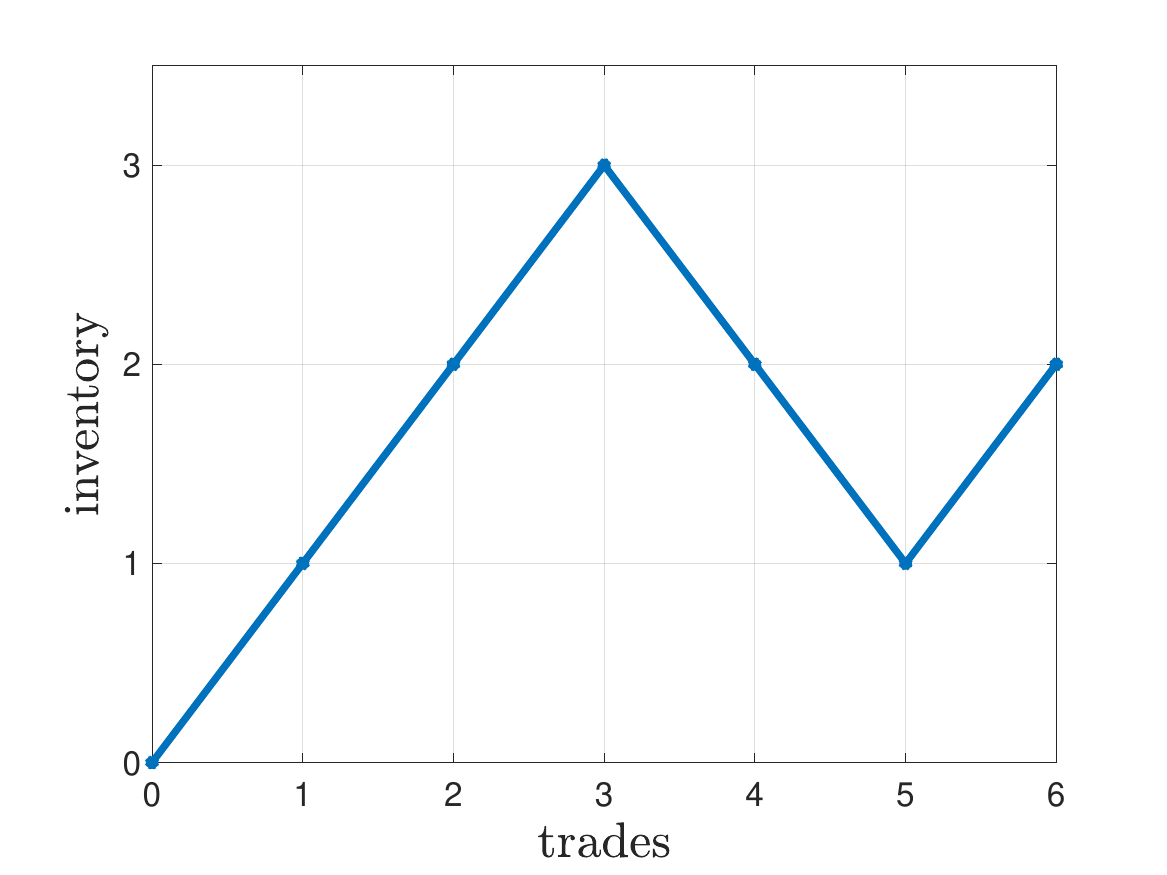}
        \caption{The total profitability of Trader $A$ is illustrated in this example involving six trades, where
$\epsilon_1=1$, $\epsilon_2=1$, $\epsilon_3=1$, $\epsilon_4=-1$, $\epsilon_5=-1$, $\epsilon_6=1$. A line leading up to a trading point indicates a buy at that point and a line leading down indicates a sell. The inventory rises to 3 at trade 3, then drops to 1 at trade 5, then rises to 2 at trade 6. Thus $\alpha_1=4$, $\alpha_2=5$, $\beta_1=3$, $\beta_2=2$. The value of Trader $A$'s position is given by the risk-neutral probability of the trade sequence multiplied by $(1-\phi^{-1}) + (1-\phi^{-2}\psi^{-1}) + (1-\phi^{-3}\psi^{-2}) - (1-\phi^{-2}\psi^{-3})   - (1-\phi^{-1}\psi^{-2}) + (1-\phi^{-2}\psi^{-1})$. We observe that since $\phi > \psi \geq 1$ the difference between the third trade and the fourth trade is positive, and likewise the difference between the second trade and the fifth trade is positive. } 
    \label{fig: lemma 3}
\end{figure} 
\noindent Define $\Lambda = \{1,\dots,n\} \setminus \left(\{\gamma_1,\dots,\alpha_{N}\}\cup \{\beta_1,\dots,\beta_{N}\}\right)$. Then
\begin{align} \label{trading terms}
    \sum_{k=1}^{n} \epsilon_k\left(1 - \phi^{-Q_k}\psi^{-Q_{k-1}} \right)
     = \sum_{k\in\Lambda} \left(1 - \phi^{-Q_k}\psi^{-Q_{k-1}} \right) + \sum_{i=1}^{N} \left(\phi^{-Q_{\alpha_i}} \psi^{-Q_{\alpha_{i}-1}} - \phi^{-Q_{\beta_i}} \psi^{-Q_{\beta_{i}-1}}  \right) ,
\end{align}
which is strictly greater than zero. In particular, we note that the first term of the right hand side of \eqref{trading terms} is strictly positive, since $Q_1 = 1$ and $\phi > \psi \geq 1$. Furthermore, since by construction $Q_{\beta_i} - Q_{\alpha_i} \geq 1$ for $i = 1, \dots, N$, we have $Q_{\beta_i-1} - Q_{\alpha_i-1} \geq -1$, and hence
\begin{align}
\phi^{-Q_{\alpha_i}}& \psi^{-Q_{\alpha_{i}-1}} - \phi^{-Q_{\beta_i}} \psi^{-Q_{\beta_{i}-1}} \nonumber \\
&= \phi^{-Q_{\beta_i}} \psi^{-Q_{\beta_{i}-1}}(  \phi^{Q_{\beta_i} - Q_{\alpha_i} } \, \psi^{-Q_{\beta_{i}-1} -Q_{\alpha_{i}-1}} - 1) 
\geq \phi^{-Q_{\beta_i}} \psi^{-Q_{\beta_{i}-1}} (  \phi \,\psi^{-1} - 1)\, ,
\end{align}
from which we deduce that the second term on the right side of \eqref{trading terms} is strictly positive. %
\end{proof}

\noindent Figure 4 provides a visual representation of the inequality \eqref{eq: lemma ineq} for sequences $(\epsilon_k)_{1\leq k\leq m}$ up to $m=10$.

\vspace{.2cm} 
\begin{figure}[H]
    \centering
    \includegraphics[height=2.2in]{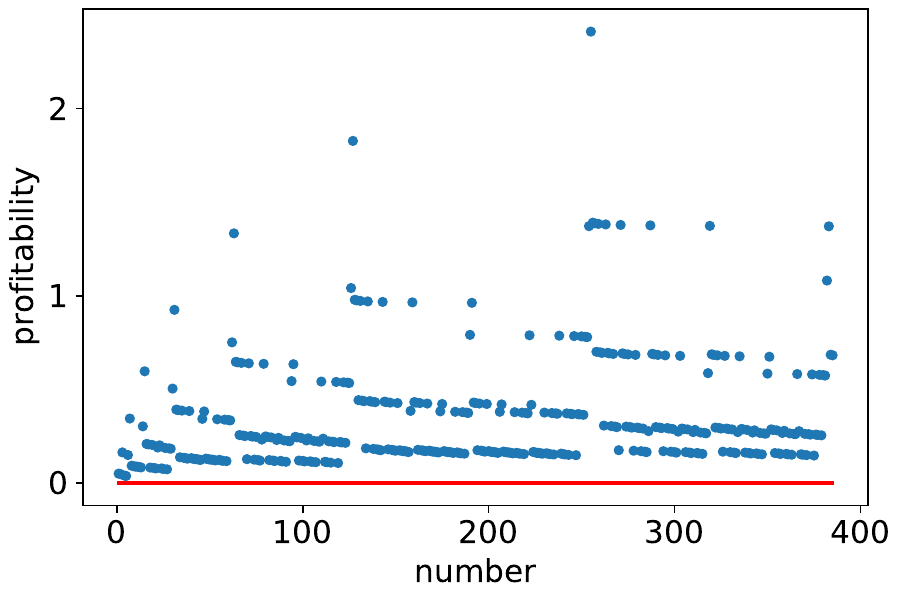}
      \caption{Positivity of profitability for $\phi=1.02$ and $\psi_A=1.01$. The positivity of  \eqref{eq: lemma ineq} is shown for these parameters and for all combinations of buys and sells up to a total of ten trades ($m\leq10$).  For $n\geq2$ we consider the binary representation of $n$ given by the binary number $1\upsilon_1\,\upsilon_2 \cdots \upsilon_m$ and we plot the value of \eqref{eq: lemma ineq} for the sequence $(\epsilon_k)_{1\leq k\leq m}$ where $\epsilon_k=2\,\upsilon_k-1$ for $k\in\{1,\dots,\,m\}$. }
    \label{fig: lemma three sum plot}
\end{figure} 
\noindent Armed with Lemma \ref{Lemma 3} we are now in a position to assert the following: 
\begin{Proposition}\label{prop: multiple trades}
The value of Trader A's position under Scenario 6 is strictly positive in any non-trivial trading model with inventory risk. 
\end{Proposition} 
\begin{proof}
We wish to show that the expectation of the total profitability \eqref{eq: n trades H^A_T} is positive for any $n\in\mathbb{N}$.  By use of the optional stopping theorem we observe that the value of Trader $A$'s position under Scenario 6 is strictly positive if 
\begin{equation}\label{eq: condition prop 9 lemma 3}
    \sum_{k=1}^{n} \epsilon_k\left(1 - \phi^{-\sum_{r=1}^k \epsilon_r}\psi^{-\sum_{r=1}^{k-1} \epsilon_r} \right) > 0
\end{equation}
for any choice of the parameters $\phi> \psi \geq 1$ and for any sequence $(\epsilon_k)_{k = 1,2,\dots,n}$ such that $\epsilon_k\in\{1,-1\}$ for all $k\in\{1,\dots, n\}$.  We conclude the proof by use of Lemma \ref{Lemma 3}.
\end{proof}
With the framework of Scenario 4 in mind, we conduct simulation studies for the case when the game master allows up to ten trades. Figure \ref{fig: trading four trades} illustrates the trading mechanism. To make the various features of the model apparent to the naked eye, we have shown the first one-fifth of the time frame of the trading session. The lower right-hand  panel shows the trajectory of  the ratio of the quoted mid-prices, given by
\begin{equation}
S_t^A\,\phi^{-Q^A_t}\,\psi_A^{-Q^A_t} / S_t^B\,\phi^{-Q^B_t}\,\psi_A^{-Q^B_t} \, .
\end{equation}
%
\begin{figure}[H]
    \centering
    \includegraphics[height=3.4in]{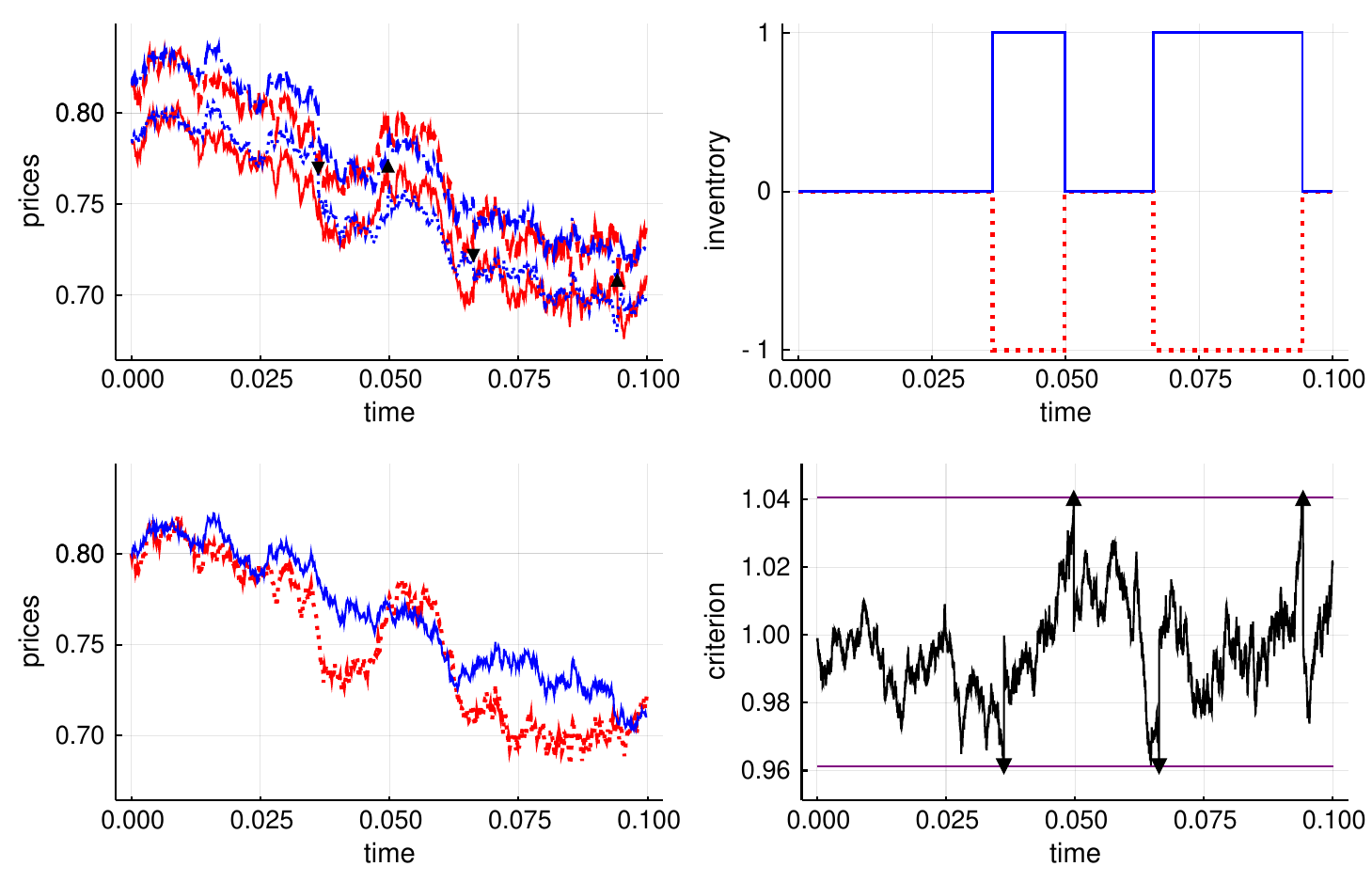}
      \caption{Trading dynamics when game master allows up to ten trades. The model parameters are $T=1$, $\sigma_B=1$, $\sigma_A=\sqrt{2}$, $\phi=1.02$, $\psi_A=1$, $\psi_B=1$ and $r=0$. A single outcome of chance is shown spanning the first one-tenth of the trading session. The top left panel shows the bid and offer quotes of Trader $A$ (red dashed line for offer, red solid line for bid) and for Trader $B$ (blue dash-dot line for offer, blue dotted line for bid). We indicate that a trade has taken place with an upward pointing arrow when Trader $A$ buys and a downward pointing arrow when Trader $A$ sells. The top right panel shows the inventories of Trader $A$ (dotted red line) and Trader $B$ (solid blue line). The bottom left panel shows the trajectories of $S^A_t$ (dotted red line) and $S_t^B$ (solid blue line). The bottom right panel shows the trajectory of  the quotient of the quoted mid-prices and the boundaries $\phi^2$ and $\phi^{-2}$. A trade occurs whenever the quotient process hits a boundary. }
    \label{fig: trading four trades}
\end{figure}

\noindent  In the figure below we plot the profitability of Trader $A$ as a function of the spread factor $\phi$, based on 100,000 simulations. The concavity of the profitability as a function of the spread is clear.
\begin{figure}[H]
    \centering
    \includegraphics[height=2.1in]{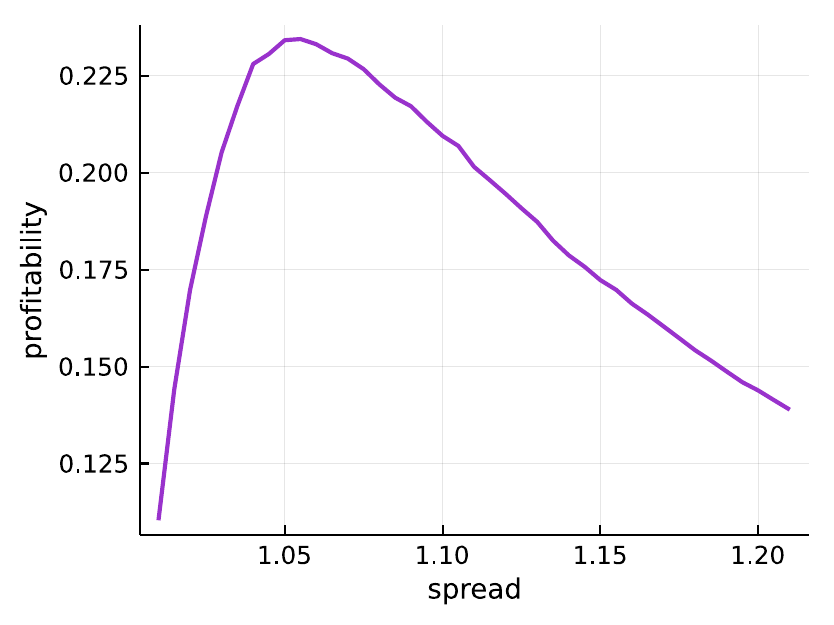}
        \caption{This plot shows the  profitability of Trader $A$ as a function of the spread. The model parameters are $r=0$, $\psi_A=1, \psi_B=1$, $\sigma_A=\sqrt{2}$, $\sigma_B=1$, and $p=0.8$. }
    \label{fig: average profitability}
\end{figure} 
\vspace{.25cm}
\noindent Figure \ref{fig: average profitability}  provides the insight that as the spread is decreased  information is transferred to the less informed trader at a lower cost. This means that there is a duality for Trader $B$, at least regarding the average profitability, between (i) acquiring more information while remaining less-informed than Trader $A$, and (ii) trading in a market where the game master declares a smaller spread. 

In Figure \ref{fig: surface profitability} we show the profitability surface as the spread factor and the information flow rate change. Figure \ref{fig: maximum inventory} demonstrates how the average inventory 
decreases as the inventory aversion parameter is rachetted up. It is interesting to observe that the dependence is essentially linear. 
\begin{figure}[H]
    \centering
    \includegraphics[height=2.0in]{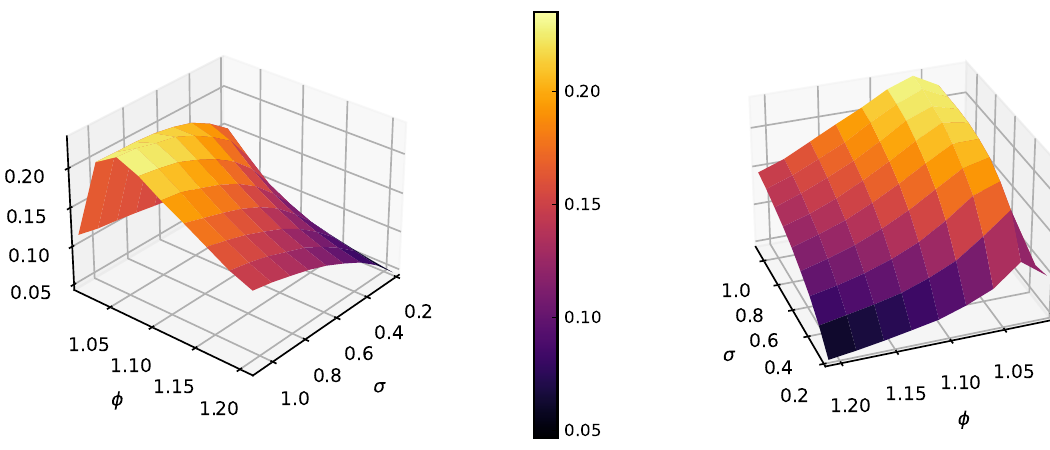}
      \caption{Trading profitability surface plotted as a function of the information flow rate $\sigma_B$ and the spread factor $\phi$. For each point on the surface we work out Trader $A$'s average profitability over the course of 100,000 simulations, with $T=1$, $\psi_A=1, \psi_B=1$, $r=0$, and $\sigma_A=\sqrt{2}\,\sigma_B$.}
    \label{fig: surface profitability}
\end{figure} 
%
%
\begin{figure}[H]
    \centering
    \includegraphics[height=2.5in]{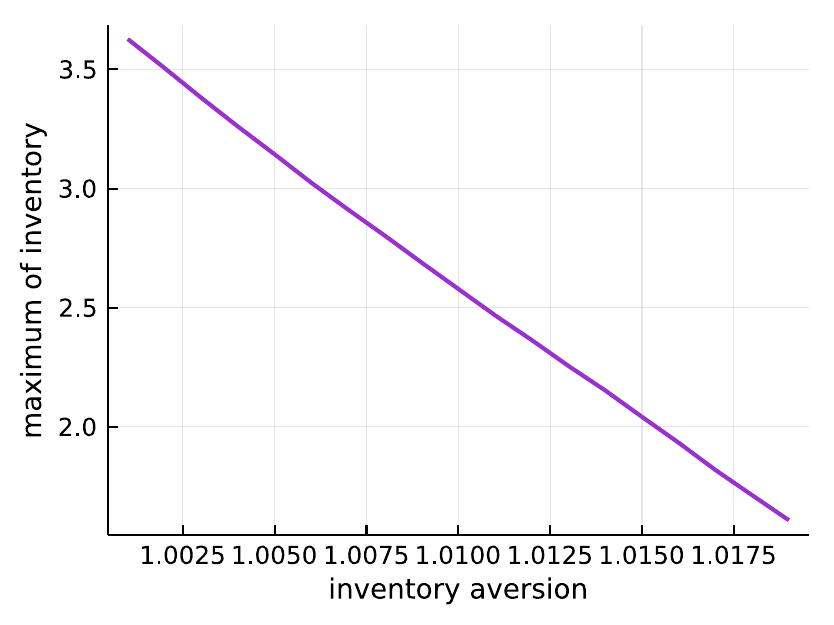}
    \caption{This plot shows the average of the maximum value of the inventory of Trader $A$ as a function of the inventory aversion parameter $\psi_A$. The model parameters are $r=0$, $\phi=1.02$, $\sigma_A=\sqrt{2}$, $\sigma_B=1$, and $p=0.8$. For a given value of $\psi_A \in [1, \phi)$ we take $\psi_B=\psi_A$. }
    \label{fig: maximum inventory}
\end{figure} 
\vspace{.5cm}
%
%
\section{Information Dominates Strategy}
\label{sec: information dominates strategy}
 \noindent  \textbf{Scenario 7}. In this scenario Trader $A$ has a fixed spread $\phi$ and a fixed ambiguity aversion parameter $\psi$ satisfying $1\leq \psi<\phi$. But Trader $B$ now employs an $\{\mcF^B_t\}_{t\geq 0}$ adapted  multiplicative spread $\{\phi_t\}_{0\leq t\leq T}$, assumed to be càdlàg, such that $\phi_t>1$ for all $t\in[0,T]$.  Trader $B$ also uses an $\{\mcF^B_t\}_{t\geq 0}$ adapted ambiguity aversion $\{\psi_t\}_{0\leq t\leq T}$, assumed to be càdlàg, satisfying $1\leq \psi_t<\phi_t$ for all $t\in[0,T]$. We define the stopping times
\begin{equation}
\tau_k^+ = \inf_{t \in [\tau_{k-1}, T]} (\phi^{-Q^A_{k-1}-1}\psi_A^{-Q^A_{k-1}}S^{A}_t \geq  \phi_t^{-Q^B_{k-1}+1}\psi_t^{-Q^B_{k-1}}S^{B}_t)\,,
\end{equation}
and
\begin{equation}
\tau_{k}^- = \inf_{t \in [\tau_{k-1}, T]} (\phi^{-Q^A_{k-1} + 1} \psi_A^{-Q^A_{k-1}}S^{A}_t \leq  \phi_t^{-Q^B_{k-1}-1}\psi_t^{-Q^B_{k-1}}S^{B}_t)\, , 
\end{equation}
where $k=1,\dots,n$, with the convention that $\tau_0=0$. Here, the inventory $Q^A_k$ of Trader $A$ after the first $k$ trades is given by 
\begin{equation}
    Q^A_k = \sum_{r=1}^k  \mathds 1(\tau_r^+ < \tau_r^-) - \mathds 1(\tau_r^- < \tau_r^+)\,,
\end{equation}
and one has $Q^B_k = - Q^A_k$. Then we obtain the following.
\begin{Proposition}\label{prop: information vs strategy}
The value of Trader A's position under Scenario 7 is strictly positive in any non-trivial trading model with inventory risk and where Trader $B$ makes use of an adapted multiplicative spread and an adapted inventory aversion. 
\end{Proposition} 
\begin{proof}
As in Proposition \ref{prop: multiple trades}, we wish to show that the expectation of Trader $A$'s profitability \eqref{eq: n trades H^A_T} is positive for any $n\in\mathbb{N}$. Given that $\{\phi_t\}_{0\leq t\leq T}$ and $\{\psi_t\}_{0\leq t\leq T}$ only enter \eqref{eq: n trades H^A_T} indirectly, via stopping times, it follows  by use of the optional stopping theorem that the value of Trader $A$'s position under Scenario 7 is strictly positive if \eqref{eq: condition prop 9 lemma 3} holds for any choice of the parameters $\phi$ and $\psi$ satisfying $1\leq \psi<\phi$ and for any sequence $(\epsilon_k)_{k = 1,2,\dots,n}$ such that $\epsilon_k\in\{1,-1\}$ for all $k\in\{1,\dots, n\}$.  Lemma \ref{Lemma 3} then applies and leads to the desired result.
\end{proof}
Proposition \ref{prop: information vs strategy} offers an interesting insight: superior information trumps strategy -- in other words, regardless of the strategy that Trader $B$ uses, the superiority of Trader $A$'s information ensures that Trader $A$'s profitability is positive. 
%
%
\section{Conclusions}
\label{sec: conclusions}
\noindent That superior information leads to superior trading need not come as a surprise, for who would have thought otherwise? But it is one thing to speculate so in general terms with waving hands, and it is another thing to embody the principle in mathematical terms underpinned with explicit models. Our examples illustrate the variety of venues under which the informed trader will practice his trade, and the professionals will realize that strategy is but of little impact,  that knowledge exceeding that of
\begin{greek}o<i pollo'i\end{greek} 
is what is required. When the traded instrument admits a single payout, as in the case of the defaultable discount bond considered, the advantage taken by the informed trader diminishes as the terminal date is approached, and at the point of reckoning traders of either class are equally knowledgeable. This is an artefact of the simple structure of the example, and not representative of the situation in general; for in reality the informed trader will have already taken his profit and moved on to the next presented opportunity. With coupon bonds and dividend paying stocks, by the time Trader $B$ has caught up with Trader $A$'s knowledge of some coupon or dividend, the intention is already on the next. We do regard it significant that the filtration accessed by Trader $B$ is a strict sub-filtration of that of Trader $A$. For we assign a value to a trader's strategy and for that we require a pricing kernel adapted to the filtration of the larger mass of the less-informed represented by Trader $B$. The change of measure thus induced allows one to work with risk-neutral pricing for all traders, whose differences in pricing are due to differences in knowledge -- not, in the main, in our scheme, to differences in  behavioral characteristics or issues of supply and demand. 

Throughout the analysis we have worked at two different levels, which we might call the general and the specific. For some considerations we look at general trading models based on a market with a pricing kernel and a hierarchy of filtrations. For other considerations we construct specific examples of market information flows. It is a general feature of the information-based approach that quantification of the magnitude of the information is left as an abstraction, and one might be left wondering what the information flow rate signifies. This issue has been addressed in references \cite{BHM2007, BHM2008}, where it is pointed out that one can back out an implied information flow rate from option prices, alongside the implied $\mathbb P$-distributions of the cash flows. Whether or not such options are actually traded is not the point; rather, we emphasize that the flow rate parameters and the {\it a priori} \,$\mathbb P$-distributions can in principle be determined from market data. We envisage a variety of practical applications of the trading scenarios that we have considered in this paper. Both in-house risk managers and market regulators may find our model useful to stress market conditions such as volatility and spread and to look at average profitability, volume traded, and other variables of interest.  
More generally, our model can be readily extended to the situation where there are more than two traders in which the traders are stratified into hierarchies. Such scenarios may form the basis for viability studies for large markets and eventually new approaches to the vexing issues of systemic risk in such markets \cite{Hurd}.
In the specialized aspects of our analysis we have confined our studies to the consideration of Brownian bridge information processes of the type set out and studied in \cite{BHM2007, BHM2008, BDFH2009, BHM2011, FHM2012, HM2012, BHM2022,  Macrina2006, Rutkowski Yu, Aydin2017, FKT2019}. This can be justified on the basis of the high degree of analytic tractability found in such models. But it should be clear that the more general aspects of our analysis extend into the categories of information processes admitting jumps, and hence it would also be worth exploring trading models based on L\'evy processes and L\'evy-Ito processes \cite{Applebaum,BHJS2021, BH2013, BHM2008dam, BHMackie2012, BHY2013, Hoyle2010, HHM2012,  HHM2015, HMM2020, HSB2020, MS2019, Menguturk 2013, Sato 1990}.

The information-based trading models that we set out here can be placed in the context of the distinct and overlapping contributions made by various authors, including those of the present paper, to the development of the trading mechanisms built on information-based asset pricing. The early work  \cite{BHM2007, BHM2008, BHM2008dam, Macrina2006, Rutkowski Yu} in this area was concerned with applications to asset pricing, derivatives risk management, and insurance markets. The story of how these collaborations came about can be found in the preface to \cite{BHM2022} where a bibliography of later work on the topic by a host of authors can be found.  The first applications to trading mechanisms appear in reference  \cite{BDFH2009}. In that paper, the trading involves (i) a large homogeneous market with access to a single flow of information, together with (ii) a single ``informed" trader, who has access to an additional flow of information. Since the informed trader also has access to the information flow available to the general market, his filtration is strictly larger than that of the market as a whole. 

In the language of the present paper, the market as a whole in \cite{BDFH2009} can be represented by a tier-1 trader, whereas the informed trader is a tier-2 trader. In reference  \cite{BDFH2009} one also finds a version of the ``Pythagorian" formula as well as elements of the statistical arbitrage argument that we use in the present paper. In reference \cite{BHM2011} one sees the first applications of the information-based  approach to trading in a truly heterogeneous market. In that work each trader has access to a distinct information process and it is assumed for simplicity that the Brownian bridges are independent. The traders make prices with bid-ask spreads and trades occur when the spreads cross. Some of the traders have informational superiority in the sense that their information flow rates are higher than those of other traders, but the relationship of the traders to one another is not ``hierarchical"  in the way that it is in  the present paper. The model of \cite{BHM2011} was extended and analyzed in great depth in the work of Ayd{\i}n \cite{Aydin2017} and has been further extended (with the inclusion of noise correlations and extensive numerical analysis) by Fukuda, Kondo \& Takada \cite{FKT2019}. All in all, one sees two distinct lines of development in the construction and analysis of trading models making use of the information-based framework. On the one hand, there are the heterogeneous models of  \cite{BHM2011, Aydin2017, FKT2019}, and on the other hand there are the hierarchical models of \cite{BDFH2009} and the present paper. These two classes of models apply to different types of markets and each can be generalized in various ways. 
An open problem in the theory of heterogeneous markets, in the situation where each trader has his own supply of information, is to show rigorously (or else provide a counterexample) that informational superiority -- that is, having a higher information flow rate -- is advantageous in the sense that it leads to a statistical arbitrage.  

\begin{acknowledgments}
\noindent
We are grateful for comments by participants at the 16th Research in Options conference (RIO 2021), where this work was presented in November 2021, co-hosted by the Department of Mathematics, Khalifa University, UAE, by the School of Applied Mathematics, Funda\c c\~ao Getulio Vargas, Rio de Janeiro (FGV EMAp), by Universidade Federal Fluminense, Brazil (UFF), and by Universidade Federal de Santa Catarina, Brazil (UFSC).  We are grateful to seminar participants in the Department of Mathematics and Statistics at McMaster University, Ontario, November 2022, for their input. The authors wish to thank A.~Macrina, and the anonymous referees, for helpful comments on earlier drafts. This work was carried out in part while LBS was based at Imperial College London and at King's College London. 
\end{acknowledgments}

\vspace{1.5cm}

\noindent {\bf References}

%
%
%
%
%
%
%
%
%
%
%
%



\end{document}